\documentclass[12pt]{article}
\usepackage{amsmath,amssymb,amsthm}
\usepackage[colorlinks,citecolor=blue,urlcolor=blue]{hyperref}
\usepackage{natbib}
\usepackage{graphicx}
\usepackage{multirow}
\usepackage{dsfont}
\usepackage{epstopdf}
\usepackage{booktabs}
\usepackage{textcomp}
\usepackage[normalem]{ulem}
\usepackage{color}
\usepackage{threeparttable}
\usepackage{comment}
\usepackage{float}
\usepackage{indentfirst}
\usepackage{subcaption}
\usepackage{arydshln}
\usepackage{cleveref}
\usepackage{graphics}
\usepackage{pifont}
\newcommand{\blind}{0}

\newtheorem{assumption}{Assumption}
\newtheorem{theorem}{Theorem}
\newtheorem{proposition}{Proposition}
\newtheorem{lemma}{Lemma}

\newtheorem{example}{Example}
\newtheorem{definition}{Definition}

\def\sumi{\sum_{i=1}^n}
\def\sumj{\sum_{j=1}^n}
\newcommand{\bs}{\boldsymbol}
\newcommand{\prob}{\mathbb{P}}
\newcommand{\mt}{\mathcal{T}}

\newcommand{\E}{\mathbb{E}}
\newcommand{\HT}{\textnormal{HT}}
\newcommand{\Fisher}{\textnormal{F}}
\newcommand{\Lin}{\textnormal{L}}
\newcommand{\NDF}{\textnormal{ND-F}}
\newcommand{\ND}{\textnormal{ND}}
\newcommand{\NDL}{\textnormal{ND-L}}

\newcommand{\opt}{\textnormal{opt}}

\newcommand{\haj}{\textnormal{Haj}}

\def\Var{\textnormal{Var}}
\def\Cov{\textnormal{Cov}}

\def\op{o_{\mathbb{P}}}
\def\Op{O_{\mathbb{P}}}

\def\cT{\mathcal{T}}

\DeclareMathOperator*{\argmin}{arg\,min}

\begin{document}

\def\spacingset#1{\renewcommand{\baselinestretch}%
{#1}\small\normalsize} \spacingset{1}


\if1\blind
{
  \title{\bf Adjusting auxiliary variables under approximate neighborhood interference}
    \date{}
  \maketitle
}\fi

\if0\blind
{
	\title{\bf Adjusting auxiliary variables under approximate neighborhood interference}
	
	\author{Xin Lu$^1$, Yuhao Wang$^{2,3}$ and Zhiheng Zhang$^2$\thanks{The authors contributed equally to this work, names are in alphabetical order. Correspondence should be addressed to XL: \url{lux20@mails.tsinghua.edu.cn}.} \\ \\
    $^1$ Department of Statistics and Data Science, Tsinghua University\\
    $^2$ Institute for Interdisciplinary Information Sciences, Tsinghua University\\
    $^3$ Shanghai Qi Zhi Institute
    }
	\date{}
	\maketitle
}\fi

\bigskip
\begin{abstract}
Randomized experiments are the gold standard for causal inference. However, traditional assumptions, such as the Stable Unit Treatment Value Assumption (SUTVA), often fail in real-world settings where interference between units is present. Network interference, in particular, has garnered significant attention. Structural models, like the linear-in-means model, are commonly used to describe interference; but they rely on the correct specification of the model, which can be restrictive. Recent advancements in the literature, such as the Approximate Neighborhood Interference (ANI) framework, offer more flexible approaches by assuming negligible interference from distant units.
In this paper, we introduce a general framework for regression adjustment for the network experiments under the ANI assumption. This framework expands traditional regression adjustment by accounting for imbalances in network-based covariates, ensuring precision improvement, and providing shorter confidence intervals. We establish the validity of our approach using a design-based inference framework, which relies solely on randomization of treatment assignments for inference without requiring correctly specified outcome models.
\end{abstract}

\noindent%
\vfill
\newpage
\spacingset{1.9} 
\section{Introduction}\label{sec:introduction}
Randomized experiments are the gold standard for causal inference. Standard randomized experiments assume the stable unit treatment value assumption  \citep[SUTVA,][]{rubin1980randomization}, namely that the treatment of one unit does not affect the outcomes of any other units. However, this assumption often fails in real settings where interference between units exists. One prominent example is interference that arises from network links, also commonly known as the network interference. In response, researchers across disciplines, such as economics, social science, and public health, have conducted randomized experiments on networks to explore interference or the spillover effect \citep{cai2015social, paluck2016changing,haushofer2018long,BEAMAN2018147, carter2021subsidies}.

The analysis of network experiments with interference typically depends on the nature of the interference in certain contexts. One common approach is to restrict interference through structural models. For instance, the well-known linear-in-means model \citep{manski1993identification, bramoulle2009identification} assumed that a unit's outcome is influenced by the average outcomes and characteristics of its network neighbors. Similarly, \cite{li2022random} considered interference arising from the treated proportion of network neighbors. In contrast, \cite{sussman2017elements} proposed a heterogeneous linear model that accounts for individual heterogeneity rather than relying solely on the treated proportion. Additionally, \cite{aronow2017estimating} introduced the concept of exposure mapping, offering a more flexible framework for representing interference through a low-dimensional function of treatment assignments.

Despite their broad applications, the requirement of accurately specified structural models can be quite stringent. This has led to a growing interest in analyzing experiments without any model specification of the exposure mapping function~\citep{savje2024causal,leung2022causal,viviano2023causal}. Within this literature, the ``Approximate Neighborhood Interference" (ANI) assumption as a popular framework, posits that interference from more distant units should be negligible \citep{leung2022causal,hoshino2023causal,gao2023causal}.

Beyond observed outcomes, baseline covariates, such as demographic variables (e.g., gender and age), are often collected before experiments. Regression adjustment is commonly applied in the analysis stage to incorporate these covariates for more accurate treatment effect estimation (i.e., reducing the asymptotic variacne) \citep{fisher1935design,Lin2013Agnostic}. While regression adjustment is well-studied under SUTVA, its application in settings with interference has received much less attention. Exceptions include \cite{ren4783803design}, which explored regression adjustment under stratified interference \citep{hudgens2008toward,imai2021causal}; and \cite{gao2023causal}, which extended Fisher’s and Lin’s regressions to the ANI settings.

In this paper, we contribute to the literature by introducing a general framework for regression adjustment under the ANI assumption. Our framework encompasses three key elements: auxiliary variables, normalizing functions applied to these variables, and network-dependent regression. These elements expand the traditional no-interference regression adjustment framework in three ways. First, the auxiliary variables broaden the concept of covariates to include variables that depend on both covariates and network structure. For example, in the linear-in-means model, one unit’s outcome is influenced by its own covariate, the covariate average, and the treated proportion of its network neighbors, all of which serve as auxiliary variables. Accounting for all of them in the regression adjustment stage is necessary when the outcome follows a linear-in-means model. However, variables such as the treated proportion of network neighbors have been overlooked by the SUTVA framework. By introducing the concept of auxiliary variables, we fill this gap.

Second, we generalize the centering procedure applied to covariates to ensure consistency in Lin’s regression-based estimator, treating it as a special case of normalization functions. We demonstrate that a provably more efficient regression-adjusted estimator cannot be achieved with simply centered covariates. To resolve this, we propose an a new normalization procedure applied to auxiliary variables, aimed at enhancing precision for downstream analysis. The pursuit of precision improvement in regression adjustment dates back to \cite{freedman2008regression}'s critique that Fisher's regression may compromise asymptotic precision. Responding to this critique, \cite{Lin2013Agnostic} proposed Lin's regression, which guarantees precision improvement, a property that has become a central focus in regression adjustment literature \citep{negi2021revisiting}.

Finally, in order to ensure that our proposed normalizing procedure can yield an estimator that is provably more efficient than without adjustment, we propose a network dependent regression adjustment  that minimizes the Heteroscedasticity and Autocorrelation Consistent (HAC) variance estimator \citep{newey1987simple} to account for the network dependencies. This approach guarantees shorter confidence intervals and, when combined with our proposed normalizing procedure, further improves the precision. This network-dependent regression adjustment aligns with literature that formulates regression adjustments by minimizing the estimated variance \citep{ding2024first, li2020rerandomization, cochran1977sampling}. 

The asymptotic results is based on the design-based inference framework which is prevalent in the causal inference literature \citep[e.g.][]{Lin2013Agnostic,bloniarz2016lasso, li2017general,aronow2017estimating,leung2022causal}. By relying only on randomization in the treatment assignments, design-based inference provide a valid inference of treatment effects without requiring additional model or distributional assumptions on the potential outcomes. 
The remainder of the paper is organized as follows. In~\Cref{sec:framework}, we establish the design-based inference framework, introduce the ANI assumption, and review the existing regression adjustment method literature. In~\Cref{ND}, we present the network-dependent regression adjustment and discuss its asympototic properties. In~\Cref{general}, we introduce the concept of auxiliary variables and our new normalizing function, show how to combine the network-dependent regression adjustment approach introduced in~\Cref{ND} with the new variables and the normalization procedure to boost the efficiency. In~\Cref{sec:experiment}, we conduct a simulation and real data analysis to evaluate the finite-sample performance of the proposed methods. In~\Cref{sec:conclusion}, we conclude the paper by a discussion of future research.

 \section{Framework, notations and literature review}\label{sec:framework}
\subsection{Framework}   
 We consider a finite population of $n$ units, $\mathcal{N}_n = \{1,\ldots,n\}$. We observe a network structure among units, which is described by an adjacency matrix $\bs{A} = (A_{ij})_{i,j\in \mathcal{N}_n}$ where $A_{ij}\in \{0,1\}$. $A_{ij}=1$ suggests that $i$ and $j$ are connected. We write the set of all possible networks by $\mathcal{{A}}_n$, and write $\bs{1}_n$ as the vector of all ones of length $n$. For simplicity, we assume the network is undirected and has no self-links so that $\bs{A}$ is symmetric and $A_{ii}=0$ for $i\in\mathcal{N}_n$. We consider an experiment with binary treatment, and let the random variable $D_i\in \{0,1\}$ denote the treatment assigned to individual $i$, whose distribution is known a priori. Moreover, we assume that the random variables {$\{D_i\}_{i\in \mathcal{N}_n}$ are mutually independent.}
 
 Let $Y_i$ be the observed outcome of unit $i$ after the experiment and $\bs{Y} := (Y_1,\ldots,Y_n)^\top$. Given $\bs{d}\in \{0,1\}^n$, we define $Y_i(\bs{d})$ as the potential outcome of unit $i$ if the population is assigned with treatment $\bs{d}$. By construction, $Y_i = Y_i(\bs{D})$. Apparently, here we assume the potential outcomes depend not only on one's own treatment, but also that of others, thus violating the classical SUTVA assumption \citep{imbens2015causal}.

Since we only observe one realization in the $2^n$ potential outcomes for each unit, it is, in general, impossible to define identifiable causal estimand without further assumptions. The predominant method is to define causal estimand through a low-dimensional function called exposure mapping~\citep{aronow2017estimating}. The exposure mapping is a pre-specified function $\bs{T}:\mathcal{N}_n\times \{0,1\}^n\times \mathcal{A}_n\rightarrow \mathcal{T}$ where $\mathcal{T} \subseteq \mathbb{R}^{d_{\bs{T}}}$ for some $d_{\bs{T}} \in \mathbb{N}$ is a finite set. With a slight abuse of notation, we denote its realization for unit $i$ as $\bs{T}_i := \bs{T}(i,\bs{D},\bs{A})$. 

Equipped with $\bs{T}_i$, we may define the expected outcome of $i$ under exposure mapping value $\bs{t}$ as~\citep{leung2022causal}:
\[
\mu_i(\bs{t}) = \sum_{\bs{d}\in \{0,1\}^n} \prob(\bs{D} = \bs{d}| \bs{T}_i=\bs{t}) Y_i(\bs{d}).
\]
 Let $\mu(\bs{t}) =\sum_{i=1}^n \mu_i(\bs{t})/n$ be its finite-population average.
We define the causal estimand as the change of the expected average potential outcomes when changing $\bs{T}_i$ from $\bs{t}$ to $\bs{t}^\prime$: $\tau(\bs{t},\bs{t}^\prime) = \mu(\bs{t})-\mu(\bs{t}^\prime)$ \citep{leung2022causal,gao2023causal}. For instance, we can define the exposure mapping  by  $\bs{T}_i = (D_i, \sum_{j=1}^n D_j A_{ij}/\sum_{j=1}^n A_{ij})$, where the first component captures the direct effect of the treatment; and the second component captures the spillover effect through the treated proportion of neighbors. In the rest of this paper, for simplicity, we suppress the dependency of the variables on the pair $(\bs{t},\bs{t}^\prime)$ when the context is clear, such that $\tau_i$ and $\tau$ are the individual and average treatment effects, respectively. Also, since we consider an asymptotic regime where the population size $n$ goes to $\infty$ and the randomness comes only from treatment assignment $\bs{D}$, all quantities, including network structure $\bs{A}$ depend implicitly on $n$. We suppress their dependency on $n$ for simplicity when the context is clear.
 
\subsection{Finite-population average treatment effect estimation under network interference without covariate information}
\label{sec:estimator-without-covariate-information}

Just as the literature for finite-population average treatment effect estimation with SUTVA assumption, under the context of network interference, there are two famous estimators for $\tau$: the Horvitz-Thompson estimator and the Hajek estimator. Given any $\bs{t} \in \mathcal{T}$, let $\pi_i(\bs{t}) = \mathbb{P}(\bs{T}_i = \bs{t})$ for $i=1,\ldots,n$, the Horvitz-Thompson (HT)  estimator \citep{leung2022causal} is defined as the inverse probability weighting of the outcomes: 
\begin{equation*}\begin{aligned}
\hat{\tau}_\HT := \hat{\mu}_{\operatorname{HT}}(\bs{t}) - \hat{\mu}_{\operatorname{HT}}(\bs{t}'), \bs{t}, \bs{t}' \in \mathcal{T}, \text{where~} \hat{\mu}_{\operatorname{HT}}(\bs{t}) := \frac{1}{n} \sum_{i=1}^n\frac{{Y}_i \bs{1}(\bs{T}_i = \bs{t})}{\pi_i(\bs{t})}.
\end{aligned}
\end{equation*}
The HT estimator has the advantage of being unbiased, as long as all the $\pi_i$'s are not equal to $0$ or $1$. In the large sample limit, \citet{leung2022causal} recently further proved that the above estimator is asymptotically consistent under the following five regularity assumptions on the interference structure and potential outcomes, which will also be used in this work.

The first two assumptions are standard regularity conditions on potential outcomes and propensity score. Similar condition has also been assumed in the literature without interference.
\begin{assumption}[Overlap]
	\label{a:overlap}
	There exists constants $\underline{\pi},\overline{\pi}\in (0,1)$, such that $\pi_i(\bs{t}) \in [\underline{\pi},\overline{\pi}]$ for all $i\in \mathcal{N}_n$ and $\bs{t}\in \mathcal{T}$.
\end{assumption}
Assumption~\ref{a:overlap} requires the probability of exposure mapping value bounded away from $0$. Assumption~\ref{a:overlap} depends on the distribution of $\bs{D}$, the structure of the network and the specification of $\bs{T}_i$. For example, Assumption~\ref{a:overlap} is usually violated when we define exposure mapping by $\bs{T}_i \equiv \bs{1}(\sum_{j=1}^n A_{ij} D_j=0)$ and $\sum_{j=1}^n  A_{ij} \rightarrow \infty$ as $n$ grows.
\begin{assumption}[Bounded outcome]
	\label{a:bounded-outcome}
	There exists some constant $0 <c_Y<\infty$, such that $|Y_i(\bs{d})|<c_Y$ for all $i\in \mathcal{N}_n$ and $\bs{d}\in \{0,1\}^n$.
\end{assumption}

The third assumption is about the exposure mapping function. To formally describe it, given any $i, j \in \mathcal{N}_n$, we write $\ell_{\bs{A}}(i,j)$ be the path distance between $i$ and $j$, specifically, $\ell_{\bs{A}}(i,i)=0$. Let $\mathcal{N}(i,s;\bs{A}) := \{j\in \mathcal{N}_n: \ell_{\bs{A}}(i,j)\leq s\}$ denote $i$'s $s$-neighborhood. Let $\bs{d}_{\mathcal{N}(i,s;\bs{A})} = (d_j)_{j\in \mathcal{N}(i,s;\bs{A})}$ and $\bs{A}_{\mathcal{N}(i,s;\bs{A})}= (A_{kl})_{k,l\in \mathcal{N}(i,s;\bs{A})}$ be the subvector and submatrix of $\bs{d}$ and $\bs{A}$ restricted on $\mathcal{N}(i,s;\bs{A})$, respectively. Then we have

\begin{assumption}[Exposure mapping]
	\label{a:bounded-exposure-mapping}
	There exists some $K \in\mathbb{N}$ such that for all $i\in\mathcal{N}_n$, $\bs{A},\bs{A}^\prime \in \mathcal{A}_n,$ and $\bs{d},\bs{d}^\prime\in \{0,1\}^n$, it holds that $\bs{T}(i,\bs{d},\bs{A}) = \bs{T}(i,\bs{d}^\prime,\bs{A}^\prime)$ if $\mathcal{N}(i,K;\bs{A}) = \mathcal{N}(i,K;\bs{A}^\prime)$, $\bs{d}_{\mathcal{N}(i,K;\bs{A})} = \bs{d}^\prime_{\mathcal{N}(i,K;\bs{A}^\prime)}$ and $\bs{A}_{\mathcal{N}(i,K;\bs{A})} = \bs{A}^\prime_{\mathcal{N}(i,K;\bs{A}^\prime)}$.
\end{assumption}
Informally, Assumption~\ref{a:bounded-exposure-mapping} requires that the exposure mapping of each unit depends only on unit's own $K$-neighborhood. This restricts the causal estimand we are interested in.

The last two assumptions are also called the the approximate neighborhood interference (ANI) model assumption, which was original proposed by~\citet{leung2022causal}. Intuitively, these assumptions state that the spillover effects from other units further apart should be negligible.
\begin{assumption}[ANI]
	\label{a:ANI}
	Let 
	\begin{align*}
		\theta_{n,s} := \max_{i\in\mathcal{N}_n}\mathbb{E}[|Y_i(\bs{D})-Y_i(\bs{D}^{(i,s)})|],  
	\end{align*}
where $\bs{D}^{(i,s)}$ is random vector with $j$-th coordinate being equal to $D_j$ if and only if  $j\in \mathcal{N}(i,s;\bs{A})$, and equal to the $j$-th coordiate of $\bs{D}$'s i.i.d. copy $\bs{D}^\prime$ otherwise. Then we have
	$\sup_n \theta_{n,s}\rightarrow 0$, as $s\rightarrow \infty$.
\end{assumption}

\begin{assumption}[Weak dependency for LLN]
	\label{a:weak-dependency-for-LLN}
	Let $M_n^\partial(s) = n^{-1}\sum_{i=1}^n|\mathcal{N}^\partial(i,s;\bs{A})|$, where $\mathcal{N}^\partial(i,s;\bs{A}) = \{j:\ell_{\bs{A}}(i,j) = s\}$ is the set of units exactly $s$ path distance from $i$; let $\tilde{\theta}_{n,s} = \theta_{n,\lfloor s/2\rfloor}\bs{1}\{s>2\max\{K,1\}\}+\bs{1}\{s\leq2\max\{K,1\}\}$ where $\lfloor \cdot \rfloor$ is the floor function and $K$ is as in Assumption~\ref{a:bounded-exposure-mapping}. Then
	$\sum_{s=0}^n M_n^\partial (s)\tilde{\theta}_{n,s}=o(n)$.
\end{assumption}

With the above and a few other regularity conditions that will be mentioned below, \citet{leung2022causal} prove that the HT estimator is asymptotically normal; and the variance can be consistently estimated in the large sample limit. For more information about these assumptions, we refer the readers to \citet{leung2022causal}.

Besides HT estimator, \citet{gao2023causal} shows that the Hajek estimator can also be used for consistent and asymptotically normal estimation of population average treatment effects:
\begin{equation*}
	\begin{aligned}
		\hat{\tau}_{\operatorname{Haj}} := \hat{\mu}_{\operatorname{Haj}}(\bs{t}) - \hat{\mu}_{\operatorname{Haj}}(\bs{t}'), \bs{t}, \bs{t}' \in \mathcal{T}, \text{where~} \hat{\mu}_{\operatorname{Haj}}(\bs{t}) := \frac{\sum_{i=1}^n {Y}_i \bs{1}(\bs{T}_i = \bs{t})/\pi_i(\bs{t})}{ \sum_{i=1}^n \bs{1}(\bs{T}_i = \bs{t})/\pi_i(\bs{t})}.
	\end{aligned}
\end{equation*}

Compared to HT estimator, although the Hajek estimator is not unbiased anymore, its asymptotic variance is usually smaller than the former. Specifically, \citet{gao2023causal} proved that the asymptotic variance of Hajek esitmator is equal to 
\begin{equation*}
	\sigma_\haj^2 = \Var(n^{-1/2}\sum_{i=1}^n V_{\haj,i}) \quad\textrm{with}\quad V_{\haj,i} = \frac{({Y}_i-{\mu}(\bs{t})) \bs{1}(\bs{T}_i = \bs{t})}{\pi_i(\bs{t})} - \frac{({Y}_i-{\mu}(\bs{t}^\prime))\bs{1}(\bs{T}_i = \bs{t}^\prime)}{\pi_i(\bs{t}^\prime)},
\end{equation*}
which is usually smaller than that of the HT estimator~\citep{leung2022causal} that is equal to:
\[
\sigma_\HT^2 = \Var(n^{-1/2}\sum_{i=1}^n V_{\HT,i}) \quad\textrm{with}\quad V_{\HT,i} = \frac{{Y}_i \bs{1}(\bs{T}_i = \bs{t})}{\pi_i(\bs{t})} - \frac{{Y}_i\bs{1}(\bs{T}_i = \bs{t}^\prime)}{\pi_i(\bs{t}^\prime)}-({\mu}(\bs{t})-{\mu}(\bs{t}^\prime)).
\]
Moreover, the Hajek estimator is also popular due to its invariance under location shift and stability \citep{gao2023causal}.

\subsection{Regression adjustment for average treatment effect estimation with interference}
\label{sec:review-of-Gao-and-Ding-regression}

Under the SUTVA model, regression adjustment is usually used to improve the precision of average treatment effect estimation~\citep{freedman2008regression,Lin2013Agnostic,lei2021regression,lu2023debiased}. With the existence of interference, \citet{gao2023causal} propose two regression adjustment methods via \emph{weighted} least square regression: Fisher's regression based estimator and Lin's regression based estimator. Let $\bs{X}_i \in \mathbb{R}^p$ be the \emph{centered} covariate for unit $i$, i.e., $\bar{\bs{X}} = 0$ where $\bar{\bs{X}}$ is the empirical average of $\bs{X}_i$'s; let $\bs{X} := (\bs{X}_1, \ldots, \bs{X}_n)^\top \in \mathbb{R}^{n \times p}$ be the covariate matrix; let $\omega_i := 1 /{\pi_i(\bs{T}_i)}$ be the weight of unit $i$ for regression; Fisher's regression based estimator is defined as $\hat{\tau}_{\Fisher} := \hat{\alpha}_{\Fisher}(\bs{t})-\hat{\alpha}_{\Fisher}(\bs{t}^\prime)$, where $\hat{\alpha}_{\Fisher}(\bs{t}), \hat{\alpha}_{\Fisher}(\bs{t'})$ are computed according to the following convex program: 
\begin{align*}
\{\hat{\alpha}_{\Fisher}(\tilde{\bs{t}}), \tilde{\bs{t}}\in\mathcal{T}\},  \hat{\bs{\beta}}_{\Fisher} :=
\argmin_{\{\alpha(\tilde{\bs{t}}), \tilde{\bs{t}}\in\mathcal{T}\},\bs{\beta}}~\sum_{i \in \mathcal{N}_n}\omega_i\Big(Y_i-\alpha(\bs{T}_i)-\bs{X}_i^\top\bs{\beta}\Big)^2.
\end{align*}
Another estimator is Lin's regression based estimator~\citep{gao2023causal}, which is defined as $\hat{\tau}_{\Lin} := \hat{\alpha}_{\Lin}(\bs{t})-\hat{\alpha}_{\Lin}(\bs{t}^\prime)$, where 
\begin{align*}
 \{\hat{\alpha}_{\Lin}(\tilde{\bs{t}}), \hat{\bs{\beta}}_{\Lin} (\tilde{\bs{t}}),\tilde{\bs{t}} \in \mathcal{T}\} := \argmin_{\{\alpha(\tilde{\bs{t}}),\bs{\beta}(\tilde{\bs{t}}),\tilde{\bs{t}} \in \mathcal{T}\}}\sum_{i \in \mathcal{N}_n}~\omega_i\Big(Y_i-\alpha(\bs{T}_i)-\bs{X}_i^\top\bs{\beta}(\bs{T}_i)\Big)^2.
\end{align*}

\cite{gao2023causal} demonstrates that under additional regularity conditions, $\hat{\tau}_{\dagger}$ for $\dagger \in \{\Fisher, \Lin\}$ remain asymptotically normal estimators for $\tau$. However, their asymptotic variances are not provably smaller than without using covariate adjustments. In fact, we can give a counterexample that shows that $\hat{\tau}_{\dagger}$ may hurt the precision compared to the unadjusted estimator. We give a constructive setting when the asymptotic variance of Fisher's and Lin's regression adjusted estimator is strictly less efficient than the unadjusted estimator in the Supplementary Material.

In the next section, we propose a network-dependent regression adjustment, which guarantees variance reduction compared to that without covariate adjustment.

\section{Network-dependent regression adjustment}\label{ND}


\subsection{Ideally optimal linear adjusted estimators}

As we have discussed in Section~\ref{sec:review-of-Gao-and-Ding-regression}, the regression adjusted estimators defined by~\citet{gao2023causal} are not necessarily guaranteed to improve efficiency compared to without performing regression adjustment at all; the next important question then, is how to propose a provably more efficient regression adjustment estimator. 
To motivate our new estimator, we first introduce a broader class of estimators that takes Hovits-Thompson, Hajek estimators as well as Fisher's and Lin's estimators as special cases. 

For notation simplicity, we first redefine HT and Hajek estimators as weighted averages: $\hat{\tau}_\star  = \frac{1}{n}\sum_{i=1}^n w_{\star,i} Y_i$, $\star = \{\HT,\haj\}$, where $w_{\HT,i}   = \bs{1}(\bs{T}_i = \bs{t})/\pi_i(\bs{t}) - \bs{1}(\bs{T}_i = \bs{t}^\prime)/\pi_i(\bs{t}^\prime)$ and
\begin{align*}
    w_{\haj,i} = \frac{\bs{1}(\bs{T}_i = \bs{t})/\pi_i(\bs{t})}{ n^{-1}\sum_{i=1}^n \bs{1}(\bs{T}_i = \bs{t})/\pi_i(\bs{t})} - \frac{ \bs{1}(\bs{T}_i = \bs{t}^\prime)/\pi_i(\bs{t}^\prime)}{n^{-1}\sum_{i=1}^n \bs{1}(\bs{T}_i = \bs{t}^\prime)/\pi_i(\bs{t}^\prime)}. 
\end{align*}

 With this, and write $\bs{\beta}_{\mathcal{T}} := \{\bs{\beta}(\tilde{\bs{t}}),\tilde{\bs{t}}\in\mathcal{T}\}$ as the set of regression coefficients indexed by $\mathcal{T}$, we define 
\[
\hat{\tau}_{\star}\big(\bs{\beta}_{\mathcal{T}}\big) := n^{-1}\sum_{i=1}^n\big\{Y_i-\bs{X}_i^\top \bs{\beta}(\bs{T}_i)\big\}w_{\star,i}.\]
With a slight abuse of notation, when all the $\bs{\beta}(\tilde{\bs{t}})$'s in $\bs{\bs{\beta}}_{\mathcal{T}}$ are the same (just like Fisher's estimator), we denote it by $\hat{\tau}_{\star}(\bs{\beta})$.

Apparently, $\hat{\tau}_{\haj}$ and $\hat{\tau}_\HT$ can be regarded as $\hat{\tau}_\star\equiv \hat{\tau}_{\star}(\bs{0})$, $\star \in \{\haj, \HT\}$. $\hat{\tau}_\Fisher$ and $\hat{\tau}_\Lin$ 
fall into these categories with $\hat{\tau}_\Fisher = \hat{\tau}_{\haj}(\hat{\bs{\beta}}_\Fisher)$ and $\hat{\tau}_\Lin = \hat{\tau}_{\haj}(\hat{\bs{\beta}}_{\Lin, \mathcal{T}})$, where $\hat{\bs\beta}_{\Lin, \mathcal{T}} := \{\hat{\bs\beta}_{\Lin}(\tilde{\bs{t}}), \tilde{\bs{t}} \in \mathcal{T}\}$.
With this, it is obvious that the asymptotic variance of $\sqrt{n}(\hat{\tau}_{\star}(\bs{\beta}_{\mathcal{T}}) - \tau)$ is equal to $\Var(n^{-1/2} \sum_{i=1}^n V_{\star, i}(\bs{\beta}_{\mathcal{T}}))$, where $V_{\star, i}(\bs{\beta}_{\mathcal{T}}))$ is the same as the $V_{\star, i}$ in~\Cref{sec:estimator-without-covariate-information}, but with $Y_i$ replaced by $Y_i - \bs{X}_i^\top \bs{\beta}(\bs{T}_i)$.

In this case, it is straightforward that an optimal regression adjustment coefficient should be the one that minimizes $\Var(n^{-1/2} \sum_{i=1}^n V_{\star, i}(\bs{\beta}_{\mathcal{T}}))$, i.e., 
\begin{align*}
    \tilde{\bs{\beta}}_{\star, \mathcal{T}}^\opt := \argmin_{\bs{\beta}_{\mathcal{T}}} \Var(n^{-1/2} \sum_{i=1}^n V_{\star, i}(\bs{\beta}_{\mathcal{T}})),
\end{align*}
or 
\[
\tilde{\bs{\beta}}_{\star}^\opt := \argmin_{\bs{\beta}} \Var(n^{-1/2} \sum_{i=1}^n V_{\star, i}(\bs{\beta})),
\]
when we restrict the regression coefficients associated with different $\bs{t}$'s to be the same just like Fisher's estimator. Though these regression coefficients can induce provably more efficient estimators than without covariates, they cannot be consistently estimated from the realized experiment. This requires us to propose new statistical approaches to estimate these quantities through real data, which we answer in the next subsection.

\subsection{A new variance estimator for network-dependent regression adjustment}\label{sec:varest}


As mentioned in the previous section, to successfully estimate $\tilde{\bs{\beta}}_{\star, \mathcal{T}}^\opt$ and $\tilde{\bs{\beta}}_{\star}^\opt$, an important first step is of course to estimate the variances of $\hat{\tau}_\star(\bs{\beta}_{\mathcal{T}})$ using the observed data. Define $\hat{\bs{\mu}}_{\star,\bs{X}}(\tilde{\bs{t}})$ ($\star \in \{\HT,\haj\}, \tilde{\bs{t}} \in \mt$), analogously as $\hat{\mu}_{\star}(\tilde{\bs{t}})$ with $Y_i$ replaced with $\bs{X}_i$. Let $\hat{V}_{\star,i}(\bs{\beta}_{\mathcal{T}})$ be the empirical approximation of $V_{\star,i}(\bs{\beta}_{\mathcal{T}})$ where we replace $\mu(\bs{T}_i)$ and $\bs{\mu}_{\bs{X}}(\bs{T}_i)$ with their empirical approximations $\hat{\mu}_{\star}(\bs{T}_i)$ and $\hat{\mu}_{\star, \bs{X}}(\bs{T}_i)$, let $B_{ij} = \bs{1} (\ell_{\bs{A}}(i,j)<b_n)$, where $b_n$ is some prespecified bandwith to be chosen, we estimate the asymptotic variance of $\sqrt{n}(\hat{\tau}_{\star}(\bs{\beta}_{\mathcal{T}}) - \tau)$, which we denote by $\sigma_\star^2(\bs{\beta}_{\mathcal{T}})$, via
 \begin{align*} 
\hat{\sigma}_{\star}^2(\bs{\beta}_\mathcal{T}) := \frac{1}{n}\sum_{i=1}^n \sum_{j=1}^n  B_{ij} \hat{V}_{\star,i}(\bs{\beta}_{\mathcal{T}})\hat{V}_{\star,j}(\bs{\beta}_\mathcal{T}).
 \end{align*}
This is the HAC variance estimator used in \citet{leung2022causal} and \citet{hoshino2023causal}. We now discuss the asymptotic limits of $\hat\sigma_\star^2(\bs{\beta}_{\mathcal{T}})$ with a fixed $\bs{\beta}_{\mathcal{T}}$; to understand these limits, we need the following two assumptions.
\begin{assumption}
\label{a:bounded-X}
There exists some constant $c_{\bs{X}}>0$, such that, $\|\bs{X}_i\|_\infty \leq c_{\bs{X}}$ for all $i\in \mathcal{N}_n$.
\end{assumption}
\Cref{a:bounded-X} is the regulatity conditions on the covariates, analogous to \Cref{a:bounded-outcome} in \cite{gao2023causal}. To formally describe the second assumption, we let $M_n(s,k) = n^{-1}\sum_{i=1}^n |\mathcal{N}(i,s;\bs{A})|^k$ be the $k$-th moment of the $s$-neighborhood size and $\mathcal{J}_n(s, m)=\left\{(i, j, k, l) \in \mathcal{N}_n^4: k \in \mathcal{N}(i, m ; \bs{A}), l \in \mathcal{N}(j, m ; \bs{A}), \ell_{\bs{A}}(i, j)=s\right\}$, then we have

\begin{assumption}[Consistency of Variance Estimator]
 \label{a:consistency-for-variance-estimator}
(i) $\sum_{s=0}^n M_n^{\partial}(s) \tilde{\theta}_{n, s}^{1-\epsilon}=O(1)$ for some $\epsilon>0$, (ii) $ M_n\left(b_n, 1\right)=o\left(n^{1 / 2}\right)$, (iii) $M_n\left(b_n, 2\right)=o(n)$, (iv) $\sum_{s=0}^n\left|\mathcal{J}_n\left(s, b_n\right)\right| \tilde{\theta}_{n, s}=o\left(n^2\right)$. 
\end{assumption}

\Cref{a:consistency-for-variance-estimator} was proposed by \cite{leung2022causal}. It restricts the network structure and the rate of $\tilde{\theta}_{n,s}$.

\begin{proposition}
\label{prop:bias-of-variance-estimator-Lin-Fisher}
Under Assumption~\ref{a:overlap}--\ref{a:ANI}, \ref{a:bounded-X}, \ref{a:consistency-for-variance-estimator}, we have if we choose $b_n \rightarrow \infty$,  then for any $\star \in \{\HT,\haj\}$, 
\[
\hat{\sigma}^2_\star(\bs{\beta}_{\mathcal{T}})=\sigma^2_\star(\bs{\beta}_{\mathcal{T}})+ R(\bs{\beta}_{\mathcal{T}})+\op(1),
\]
where 
\[
R(\bs{\beta}_\mt) = \frac{1}{n}\sum_{i=1}^n\sum_{j=1}^n B_{ij} (\tau_i-\tau-\{\bs{\beta}(\bs{t})-\bs{\beta}(\bs{t}^\prime)\}^\top\bs{X}_i)(\tau_j-\tau-\{\bs{\beta}(\bs{t})-\bs{\beta}(\bs{t}^\prime)\}^\top\bs{X}_j).
\]
\end{proposition}

Apparently, in situations where all the $\bs{\beta}(\bs{t})$'s in $\mathcal{T}$ are identical to some $\bs{\beta}$, we have $R(\bs{\beta}) \equiv \frac{1}{n}\sum_{i=1}^n\sum_{j=1}^n B_{ij} (\tau_i-\tau)(\tau_j-\tau)$, i.e., that $R(\bs{\beta})$ stays the same, regardless of the choice of $\bs{\beta}$. We may simplify $R(\bs{\beta})$ as $R$ for concise. Although the bias terms $R(\bs{\beta}_\mt)$ are not necessarily non-negative, they exhibit the form of a HAC variance estimator for unit-level treatment effect. Under further conditions on the data generating process for $(\tau_i,\bs{X}_i)$, we can expect the bias to tend towards a non-negative value; refer to Theorem 4.2 of \cite{leung2019inference} for such conditions. Under SUTVA and $\bs{T}_i \equiv D_i$, if we set $b_n=0$, the bias reduces to 
\[ R(\bs{\beta}_\mt) = \frac{1}{n}\sum_{i=1}^n (\tau_i-\tau-\{\bs{\beta}(\bs{t})-\bs{\beta}(\bs{t}^\prime)\}^\top\bs{X}_i)^2,
\]
which is the well-known bias of the variance estimators of regression-adjusted estimators under completely randomized experiments. Noteworthy, by setting $\bs{\beta}_{\mathcal{T}}$ as $\bs{0}_\cT$, i.e., that all the $\bs{\beta}(\bs{t})$'s are set as zero vectors, $\sigma_\star^2(\bs{0}_\cT) + R(\bs{0}_\cT)$ is the asymptotic limit of the unadjusted variance estimators studied in~\citet{leung2022causal}.

\subsection{Network-dependent regression adjustment with new variance estimators}\label{sec:ndr}
Equipped with the above variance estimators, it is straightforward to define our network-dependent regression adjustment estimator as follows:
\begin{align*}
    &\hat{\tau}_{\star,\NDF} = \hat{\tau}_{\star}(\hat{\bs{\beta}}_{\star,\NDF}), \quad \hat{\bs{\beta}}_{\star,\NDF} := \argmin_{\bs{\beta}} \hat{\sigma}^2_{\star}(\bs{\beta});\\
    & \hat{\tau}_{\star,\NDL} = \hat{\tau}_{\star}(\hat{\bs{\beta}}_{\star,\NDL}),\quad \hat{\bs{\beta}}_{\star,\NDL} = \{\hat{\bs{\beta}}_{\star,\NDL}(\tilde{\bs{t}}),\tilde{\bs{t}}\in \mathcal{T}\} := \argmin_{\bs{\beta}_\mt} \hat{\sigma}^2_{\star}(\bs{\beta}_\mt).
\end{align*}
We use the notation $\NDL$ ($\NDF$) to denote the network-dependent regression that includes (excludes) covariate-exposure-mapping interactions, similar to Lin's (Fisher's) regression. We take $\hat{\sigma}^2_{\star}(\hat{\bs{\beta}}_{\star,\NDF})$ and $\hat{\sigma}^2_{\star}(\hat{\bs{\beta}}_{\star,\NDL})$ as the variance estimator of $n^{1/2}(\hat{\tau}_{\star,\NDF} - \tau)$ and $n^{1/2}(\hat{\tau}_{\star,\NDL} - \tau)$, respectively. We pair them to construct Wald-type confidence intervals for inference. 

Since as mentioned in~\Cref{prop:bias-of-variance-estimator-Lin-Fisher}, although $\hat\sigma_\star(\bs{\beta})$ is not a consistent estimator for $\sigma_\star(\bs{\beta})$, it has an inflation that does not depend on $\bs{\beta}$. This means that under certain conditions that we will discuss below, we have that in light of large sample,
\[
\hat{\bs{\beta}}_{\star, \NDF} = 
\argmin_{\bs{\beta}} \hat{\sigma}^2_{\star}(\bs{\beta}) \approx \argmin_{\bs{\beta}} \sigma^2_{\star}(\bs{\beta})  = \tilde{\bs{\beta}}^\opt_\star,
\]
so that $\hat\tau_{\star, \NDF}$ can achieve the same efficiency as $\hat\tau_\star(\tilde{\bs{\beta}}^\opt_\star)$. 
Regarding $\hat\tau_{\star, \NDL}$, since $R(\bs{\beta}_{\mt})$ depends on the coefficient when we do not restrict all the $\bs{\beta}(\bs{t})$'s to be identical, it is not necessarily guaranteed to achieve the same efficiency as $\hat\tau(\tilde{\bs{\beta}}^\opt_{\star,\mathcal{T}})$. Nevertheless, as we will show later, under certain regularity conditions, we can still have in the large sample limit, $\hat{\tau}_{\star,\NDL}$ is asymptotically normal. Moreover, its estimated variance is equal to 
\[
\hat{\sigma}_\star^2(\hat{\bs{\beta}}_{\star, \NDL}) = \min_{\bs{\beta}_\cT}\hat{\sigma}_\star^2(\bs{\beta}_\cT) \approx \min_{\bs{\beta}_\cT}\sigma_\star^2(\bs{\beta}_\cT) + R(\bs{\beta}_\cT) \le \sigma_\star^2(\bs{0}_\cT) + R(\bs{0}_\cT),
\]
so that $\hat\tau_{\star, \NDL}$ yields a confidence interval asymptotically shorter than the unadjusted estimators proposed by~\citet{leung2022causal}.


In the following, we provide theoretical guarantees of our network-dependent regression adjustment estimators. To study their asymptotic properties, we introduce \Cref{a:for-CLT-of-estimated-optimal-Fisher-Lin} and \ref{a:bias-same-order-with-estimator-Fisher-Lin}.  Let $\mathcal{H}_n(s,m)=\{(i,j,k,l)\in\mathcal{N}_n^4:k\in\mathcal{N}(i,m;\bs{A}),l\in\mathcal{N}(j,m;\bs{A}),\ell_{\bs{A}}\big(\{i,k\},\{j,l\}\big)=s\big\}$. Let $ \tilde{\bs{\beta}}_{\star,\NDL} := \argmin_{\bs{\beta}_\mt} {\sigma}_{\star}^2 (\bs{\beta}_\mt) + R(\bs{\beta}_\mt)$, which is the asymptotic limit of $\hat{\bs{\beta}}_{\star,\NDL}$ as we will show in the proof of \Cref{thm:estimator-with-optimal-estimated-precesion-Lin-Fisher}. 
\begin{assumption}
\label{a:for-CLT-of-estimated-optimal-Fisher-Lin}
Fix $\star \in \{\HT,\haj\}$. 
(i) 
Limit inferior of the smallest eigenvalues of the Hessian matrix of  ${\sigma}_{\star}^2(\bs{\beta}) + R(\bs{\beta})$ (with respect to $\bs{\beta}$) at $\tilde{\bs{\beta}}_{\star}^\opt$ are greater than zero; Limit inferior of the smallest eigenvalues of the Hessian matrix of  ${\sigma}_{\star}^2(\bs{\beta}_{\mt}) + R(\bs{\beta}_{\mt})$ (with respect to $\bs{\beta}_{\mt}$) at $\tilde{\bs{\beta}}_{\star,\NDL}$ is greater than zero.
(ii)
For $\dagger \in \{\tilde{\bs{\beta}}_{\star}^\opt, \tilde{\bs{\beta}}_{\star,\NDL}\}$, $\liminf_{n\rightarrow\infty}{\sigma}_{\star}^2 (\dagger)>0$
(iii) There exist $\epsilon>0$ and a sequence of positive constants $\left\{m_n\right\}_{n \in \mathbb{N}}$ such that $m_n \rightarrow \infty$, we have
$$
\max \left\{ \frac{1}{n^2} \sum_{s=0}^n\left|\mathcal{H}_n\left(s, m_n\right)\right| \tilde{\theta}_{n, s}^{1-\epsilon},  n^{-1 / 2} M_n\left(m_n, 2\right),  n^{3 / 2} \tilde{\theta}_{n, m_n}^{1-\epsilon}\right\} \rightarrow 0.
$$
\end{assumption}

\Cref{a:for-CLT-of-estimated-optimal-Fisher-Lin} (i) and (ii) are standard in the literature of regression adjustment. \Cref{a:for-CLT-of-estimated-optimal-Fisher-Lin} (i) reduces to the nonsingular condition of the finite-population covariance matrix of $\bs{X}_i$ in the SUTVA setting. \Cref{a:for-CLT-of-estimated-optimal-Fisher-Lin} (ii) ensures that the asymptotic variances are non-degenerated.  \Cref{a:for-CLT-of-estimated-optimal-Fisher-Lin} (iii) is from Assumption~6 in \cite{leung2022causal}, which requires the rate of $\tilde{\theta}_{n,s}$ to decay to $0$ fast enough and restricts the network structure. 

\begin{assumption}
    \label{a:bias-same-order-with-estimator-Fisher-Lin}
    The following $3$ terms are of order $O(1)$
    \[
    \frac{1}{n}\sumi\sumj B_{ij} (\tau_i-\tau)(\tau_j-\tau),\quad \frac{1}{n}\sumi\sumj B_{ij} (\tau_i-\tau )\bs{X}_j,\quad \frac{1}{n}\sumi\sumj B_{ij} \bs{X}_i\bs{X}_j^\top.
    \]
\end{assumption}
Under the assumptions of \Cref{prop:bias-of-variance-estimator-Lin-Fisher}, $\sigma^2_{\star}(\bs{\beta}_\mt)$ are $O(1)$. \Cref{a:bias-same-order-with-estimator-Fisher-Lin} ensures that the bias of the variance estimators, $R(\bs{\beta}_\mt)$, must have the same magnitude as the sampling variance for any adjustment coefficient. These estimators also exhibit the form of the HAC variance estimator. Under some regularity conditions on the data generating process of $(\bs{X}_i,\tau_i)$, we can expect \Cref{a:bias-same-order-with-estimator-Fisher-Lin} holds.

\begin{theorem}
\label{thm:estimator-with-optimal-estimated-precesion-Lin-Fisher}
Under Assumptions~\ref{a:overlap}--\ref{a:ANI}, \ref{a:bounded-X}--\ref{a:bias-same-order-with-estimator-Fisher-Lin}, we have (i)  $n^{1/2}(\hat{\tau}_{\star,\NDF}-\tau)/{\sigma}_{\star}(\tilde{\bs{\beta}}_{\star}^\opt)\xrightarrow{\textnormal{d}} \mathcal{N}(0,1)$, $n^{1/2}(\hat{\tau}_{\star,\NDL}-\tau)/{\sigma}_{\star}(\tilde{\bs{\beta}}_{\star,\NDL})\xrightarrow{\textnormal{d}} \mathcal{N}(0,1)$; and (ii) $\hat{\sigma}_\star^2(\hat{\bs{\beta}}_{\star,\NDF})= {\sigma}_{\star}^2(\tilde{\bs{\beta}}_{\star}^\opt)+R + \op(1)$, $\hat{\sigma}_{\star}^2(\hat{\bs{\beta}}_{\star,\NDL})= {\sigma}_{\star}^2(\tilde{\bs{\beta}}_{\star,\NDL})+R(\tilde{\bs{\beta}}_{\star,\NDL})+\op(1)$. 
\end{theorem}
 
\Cref{thm:estimator-with-optimal-estimated-precesion-Lin-Fisher}(i) demonstrate the asymptotic normality and asymptotic variance of network-dependent regression-based estimator. As discussed previously, $\hat{\tau}_{\star,\NDF}$ has the same asymptotic variance as $\hat{\tau}(\tilde{\bs{\beta}}^\opt_\star)$, i.e., that it is more efficient than without using regression adjustment at all; whilst $\hat{\tau}_{\star, \NDL}$, though still asymptotically normal, can have a variance larger than $\sigma_\star^2(\tilde{\bs{\beta}}^\opt_{\star, \mathcal{T}})$ and probably also greater than without any adjustments at all. 

\section{Regression adjustment with general auxiliary variables}\label{general}

In~\Cref{ND}, we have introduced a new regression adjustment estimator for experiments with interference. However, it still has two open questions. First, $\hat{\tau}_{\star,\NDL}$ can be less efficient than without adjustment; second, it does not effectively incorporate other pretreatment information for adjustment, such as the network topology.

In this section, we introduce general auxiliary variables, which include $\bs{X}_i$ as a special case; and a normalizing procedure, which ensures that the network dependent regression adjusted estimators are provably more efficient than unadjusted estimators.

\subsection{General auxiliary variables: definition and examples}


In SUTVA model, we perform regression adjustments on $\bs{X}_i$ to address imbalances and enhance the interpretability. However, when interference exists, 
merely adjusting for the imbalances in $\bs{X}_i$ may be insufficient. Empirical evidence suggests that the outcome of one unit may be influenced by the characteristics and treatments of its neighbors within the network. Such influence may result in an accumulation of imbalances in the observed outcomes. To develop a regression adjustment that adjusts for such imbalances, we define auxiliary variables $\bs{G}_i$ which extends $\bs{X}_i$, including a broad class of variables that may depend on $\bs{D}$ and $\bs{A}$.
 Let $\mathcal{X}_n$ be the value space of the data matrix $\bs{X}$, define auxiliary mapping as a prespecified function $\bs{G}: \mathcal{N}_n \times \mathcal{X}_n \times \{0,1\}^n\times \mathcal{A}_n\rightarrow \mathbb{R}^Q$ where $Q$ is the dimension of the projection space. With a slight abuse of notation, we let $\bs{G}_i = \bs{G}(i,\bs{X},\bs{D},\bs{A})$ be the auxiliary vector of $i$ and let $\bs{G}_i(\bs{d}) = \bs{G}(i,\bs{X},\bs{d},\bs{A})$ be its potential outcome.

Similar to the way we use $\bs{X}_i$ in the no-interference setting, our goal is to use $\bs{G}_i$ to improve the estimation precision of the average treatment effect. $\bs{G}_i$ may come from any working model, which is not necessarily correctly specified. We emphasize that, although our theory does not require $\bs{G}_i$ to be correctly specified, a well-chosen set of $\bs{G}_i$ that captures the nature of interference can significantly enhance both the precision and interpretability of the results. We end this subsection by giving two concrete examples of $\bs{G}_i$ that can be analyzed by our method.

\begin{example}[Fisher's regression and Lin's regression]
  Fisher's regression uses $\bs{G}_i\equiv\bs{X}_i$, while Lin's regression used $\bs{X}_i$ fully interacted with the exposure mapping indicator, leading to $\bs{G}_i \equiv \bs{X}_i\otimes\bs{z}_i$, where $\bs{z}_i = (\bs{1}(\bs{T}_i=\tilde{\bs{t}});\tilde{\bs{t}}\in \mathcal{T})$. Lin's regression is conceptually more reasonable if there is heterogeneity under different exposure mapping values.
\end{example}

As a by-product of the unified definition for auxiliary variables, we can view Lin's regression and Fisher's regression as the same regression but with different auxiliary variables. 

\begin{example}
\label{example:auxiliary-variable-linear-in-means}
 $\sum_{j=1}^n A_{ij}>0$ for $i \in\mathcal{N}_n$. Consider $Y_i$ generated by the following model:
\[
V_i =\alpha_0+\alpha_1\frac{\sum_{j=1}^n A_{ij}Y_j}{\sum_{j=1}^n A_{ij}}+\alpha_2 \frac{\sum_{j=1}^n A_{ij}D_j}{\sum_{j=1}^n A_{ij}}+\alpha_3 D_i+{\alpha}_4^\top \bs{X}_i+ \varepsilon_i,\quad\text{with}\quad Y_i = f(V_i),
\]
where $\alpha_j$, $j\in \{0,1,2,3\}$ are fixed scalars, $\alpha_1<1$ and ${\alpha}_4$ are fixed vector;  $\varepsilon_i$'s are random noise. For the linear-in-means model  $f(t) = t$, while for the nonlinear contagion model, $f(t) = \bs{1}(t \geq 0)$. Under this model, we can specify $\bs{G}_i \equiv (1,D_i, \sum_{j=1}^n A_{ij}D_j/\sum_{j=1}^n A_{ij},\bs{X}_i^\top)^\top$. Let $\bs{O} = (A_{ij}/\sum_{l=1}^n A_{il})_{i,j\in \mathcal{N}_n}$ and $\bs{\varepsilon} = (\varepsilon_1,\ldots,\varepsilon_n)^\top$. Under the linear-in-means model, we have $\bs{Y} = \alpha_0(\bs{I}_n-\alpha_1\bs{O})^{-1}\bs{1}_n+\alpha_2(\bs{I}_n-\alpha_1\bs{O})^{-1}\bs{O}\bs{D}+\alpha_3(\bs{I}_n-\alpha_1\bs{O})^{-1}\bs{D}+\alpha_4(\bs{I}_n-\alpha_1\bs{O})^{-1}\bs{X}+(\bs{I}_n-\alpha_1\bs{O})^{-1}\bs{\varepsilon}$. Since $(\bs{I}_n-\alpha_1\bs{O})^{-1} = \bs{I}_n + \sum_{k=1}^\infty \alpha_1^k\bs{O}^k$, we can include higher order polynomials of $\bs{O}$ such as $(\bs{O}^k\bs{D},\bs{O}^k\bs{X})$, $k\geq 2$, in set of the auxiliary variables to incorporate the higher order interference. 
\end{example}


\subsection{Towards more efficient treatment effect estimation via a new normalizing function}\label{sec:nddnorm}

Armed with the new $\bs{G}$, we renew the class of linear adjusted estimators by 
\[
\hat{\tau}_{\star}\big(\bs{\beta}\big) := n^{-1}\sum_{i=1}^n\big(Y_i-\phi(\bs{G}_i)^\top \bs{\beta}\big)w_{\star,i},\]
where $\phi$ is some normalizing procedure. Apparently, this represents a more general framework where Lin's and Fisher's estimators can be viewed as special cases. Specifically, $\hat{\tau}_\Fisher$ falls into this category by setting $\bs{G}_i$ as $\bs{X}_i$, $\bs{\beta}$ as $\hat{\bs{\beta}}_\Fisher$ and $\phi(\bs{G}_i)$ as $\bs{X}_i - \bar{\bs{X}}$; $\hat{\tau}_\Lin$ falls into this category by setting $\bs{G}_i$ as $\bs{X}_i\otimes \bs{z}_i$ where $\bs{z}_i := (\bs{1}(\bs{T}_i = \tilde{\bs{t}}),\tilde{\bs{t}}\in\mathcal{T})$, $\bs{\beta}$ as $\hat{\bs{\beta}}_{\Lin, \cT}$, and $\phi(\bs{G}_i)$ as $(\bs{X}_i - \bar{\bs{X}}) \otimes \bs{z}_i$. As discussed previously, under such choice of $\phi$, the resulting estimator is not always guaranteed to be more efficient than without adjustment. Our goal then, is to derive a new $\phi$ that guarantees both consistency and efficiency improvement of the point estimator, regardless of the choice of general auxiliary variables.

We consider bounded $\bs{G}_i$ similar to \Cref{a:bounded-X} and restrict the spillover from other units on $\bs{G}_i$. Such $\bs{G}_i$ is formalized by \Cref{a:bounded-G-dc}.
\begin{assumption}
\label{a:bounded-G-dc}
 There exists some constant $c_{\bs{G}}>0$, such that, $\|\bs{G}_i(\bs{d})\|_\infty \leq c_{\bs{G}}$ for all $i\in \mathcal{N}_n$ and $\bs{d}\in \{0,1\}^n$. Moreover, $\max_{i\in\mathcal{N}_n}\mathbb{E}\|\bs{G}_i(\bs{D})-\bs{G}_i(\bs{D}^{(i,s)})\|_{\infty} \leq c_{\bs{G}} \theta_{n,s}$.
\end{assumption}

Just as in~\Cref{sec:varest}, we define $\sigma_\star^2(\bs{\beta})=\Var(n^{-1/2} \sum_{i=1}^n V_{\star, i}(\bs{\beta}))$ as the asymptotic variance of $\hat{\tau}_{\star}\big(\bs{\beta}\big)$ and $V_{\star, i}(\bs{\beta})$ is the same as $V_{\star,i}$ but with $Y_i$ replaced by $Y_i-\phi(\bs{G}_i)^\top\bs{\beta}$. Since the definitions in~\Cref{sec:varest} can be viewed as a special case with $\phi(\bs{G}_i) = (\bs{X}_i-\bar{\bs{X}})\otimes \bs{z}_i$ under the framework of auxiliary variable, there is no abuse of notation. Analogously, we define the optimal adjustment coefficient $\tilde{\bs{\beta}}_{\star}^\opt := \argmin_{\bs{\beta}\in \mathbf{R}^Q}\sigma_\star^2(\bs{\beta})$.  We propose to estimate $\bs{\beta}$ via minimizing the variance estimator of $\hat{\tau}_{\star}\big(\bs{\beta}\big)$:
 \begin{align*}
\hat{\sigma}_{\star}^2(\bs{\beta}) := \frac{1}{n}\sum_{i=1}^n \sum_{j=1}^n  B_{ij} \hat{V}_{\star,i}(\bs{\beta})\hat{V}_{\star,j}(\bs{\beta})
 \end{align*}   
where $\hat{V}_{\star,i}(\bs{\beta})$ is an empirical approximation of $V_{\star, i}(\bs{\beta})$ with ${\mu}(\bs{T}_i)$ replaced with $\hat{\mu}_{\star}(\bs{T}_i)$ and $\bs{\mu}_{\phi(\bs{G})}(\bs{T}_i)$ replaced with $\hat{\bs{\mu}}_{\star,\phi(\bs{G})}(\bs{T}_i)$. Therefore, there is no abuse of notation. We define the network dependent estimator accordingly:
\[
 \hat{\tau}_{\star,\ND}:= \hat{\tau}_{\star}(\hat{\bs{\beta}}_{\star,\ND}), \quad \hat{\bs{\beta}}_{\star,\ND}= \argmin_{\bs{\beta}\in \mathbf{R}^Q}  \hat{\sigma}_{\star}^2(\bs{\beta}).
\]

However, as mentioned in~\Cref{thm:estimator-with-optimal-estimated-precesion-Lin-Fisher}, due to the dependence of variance inflation on the regression coefficient $\bs{\beta}$, such an estimate is not always guaranteed to result in a variance reduction. This fact is reiterated by the following parallel result of \Cref{prop:bias-of-variance-estimator-Lin-Fisher} on the bias of the variance estimator. 

\begin{proposition}
\label{prop:bias-of-variance-estimator}
Suppose we are under \Cref{a:bounded-G-dc}, but with $\bs{G}_i$ replaced by $\phi(\bs{G}_i)$. Suppose further Assumption~\ref{a:overlap}--\ref{a:ANI}, \ref{a:consistency-for-variance-estimator}, $\lim_{n\rightarrow \infty} b_n \rightarrow \infty$, then we have for any $\star \in \{\HT,\haj\}$, $\hat{\sigma}^2_{\star}(\bs{\beta})={\sigma}^2_{\star}(\bs{\beta})+R(\bs{\beta})+\op(1)$, where
\[
R(\bs{\beta}) = \frac{1}{n}\sum_{i=1}^n\sum_{j=1}^nB_{ij} \{\tau_i-\tau -\bs{\beta}^\top(\bs{\tau}_{\phi(\bs{G}),i}-\bs{\tau}_{\phi(\bs{G})})\}\{\tau_j-\tau -\bs{\beta}^\top(\bs{\tau}_{\phi(\bs{G}),j}-\bs{\tau}_{\phi(\bs{G})})\},
\]
and $\bs{\tau}_{\phi(\bs{G})} = n^{-1}\sum_{i=1}^n \bs{\tau}_{\phi(\bs{G}),i}$
\end{proposition}

As shown in the discussions in \Cref{sec:ndr}, the dependence of $R(\bs{\beta})$ on $\bs{\beta}$ is the main obstacle in deriving an optimal adjusted estimator. If we can derive $\phi$ such that $\bs{\tau}_{\phi(\bs{G}),i} = \bs{0}$, for $i\in \mathcal{N}_n$, we can remove the dependence of $R(\bs{\beta})$ on $\bs{\beta}$, yielding an optimal estimator. This idea motivates our new normalizing procedure, $\phi_{0}$. Note that $\bs{\tau}_{\bs{G},i} = \E w_{\HT,i} \bs{G}_{i}$. This is the correlation between $w_{\HT,i}$ and $\bs{G}_i$ at the individual level. We derive $\phi_{0}$ as a procedure to decorrelate $w_{\HT,i}$ and $\bs{G}_i$:
\begin{align}
\label{eq:point-wise-decorrelation}
    \phi_{0} (\bs{G}_{i}) = \bs{G}_i-\bs{\gamma}_i w_{\HT,i},\quad \bs{\gamma}_i =\frac{\E w_{\HT,i}\bs{G}_i}{\E w_{\HT,i}^2 }.
\end{align}
Apparently, $\bs{\tau}_{\phi_0(\bs{G}),i} = \bs{0}$, which satisfies our requirement. When $\bs{G}_i \equiv (\bs{X}_i^\top \bs{1}(\bs{T}_i=\bs{t}), \bs{X}_i^\top \bs{1}(\bs{T}_i=\bs{t}^\prime))^\top$ and $\pi_i(\bs{t}) + \pi_i(\bs{t}^\prime) = 1$, a straightforward calculation shows that $\phi_0(\bs{G}_i) = (\pi_i(\bs{t})\bs{X}_i^\top, \pi_i(\bs{t}^\prime)\bs{X}_i^\top)^\top$. Thus, $\phi_0(\bs{G}_i)$ replaces the exposure mapping indicators with their expected values. However, in general, the normalized variable $\phi_0(\bs{G}_i)$ does not always admit a simple closed-form expression, and we often rely on Monte Carlo methods to compute $\bs{\gamma}_i$ for $\phi_0(\bs{G}_i)$.

We now build the asymptotic normality and optimality of the network-dependent regression adjustment based on $\phi_{0}(\bs{G}_i)$. \Cref{a:for-CLT-of-estimated-optimal}  is parallel to \Cref{a:for-CLT-of-estimated-optimal-Fisher-Lin} under the framework of auxiliary variables. \Cref{a:bias-same-order-with-estimator} is part of \Cref{a:bias-same-order-with-estimator-Fisher-Lin}.
\begin{assumption} 
\label{a:for-CLT-of-estimated-optimal}
For $\star\in\{\HT,\haj\}$, We have
(i) Limit inferior of the smallest eigenvalues of the Hessian matrix of  ${\sigma}^2_{\star}(\bs{\beta})+R(\bs{\beta})$ (with respect to $\bs{\beta}$) at $\tilde{\bs{\beta}}_\star^\opt$ is greater than zero (ii) $\liminf_{n\rightarrow\infty}\sigma_\star^2(\tilde{\bs{\beta}}^\opt_{\star})>0$
(iii) There exist $\epsilon>0$ and a sequence of positive constants $\left\{m_n\right\}_{n \in \mathbb{N}}$ such that $m_n \rightarrow \infty$, we have
$$
\max \left\{ \frac{1}{n^2} \sum_{s=0}^n\left|\mathcal{H}_n\left(s, m_n\right)\right| \tilde{\theta}_{n, s}^{1-\epsilon},  n^{-1 / 2} M_n\left(m_n, 2\right),  n^{3 / 2} \tilde{\theta}_{n, m_n}^{1-\epsilon}\right\} \rightarrow 0.
$$
\end{assumption}

\begin{assumption}
    \label{a:bias-same-order-with-estimator}
 Let   $R = \frac{1}{n}\sumi\sumj B_{ij} (\tau_i-\tau)(\tau_j-\tau)$ and $R = O(1)$.
\end{assumption}

\begin{theorem}
\label{thm:estimator-with-optimal-estimated-precesion}
Under Assumption~\ref{a:overlap}--\ref{a:ANI}, \ref{a:consistency-for-variance-estimator}, \ref{a:bounded-G-dc}--\ref{a:bias-same-order-with-estimator}, with $\phi  \equiv \phi_{0}$,  We have 
(i) $n^{1/2}(\hat{\tau}_{\star,\ND}-\tau)/\sigma_\star(\tilde{\bs{\beta}}_\star^\opt)\xrightarrow{\textnormal{d}} \mathcal{N}(0,1)$ and (ii) $\hat{\sigma}_\star^2(\hat{\bs{\beta}}_{\star,\ND}) - \sigma_\star^2(\tilde{\bs{\beta}}_\star^\opt)- R = \op(1)$. 
\end{theorem}


\Cref{thm:estimator-with-optimal-estimated-precesion} (ii) shows the asymptotic limit of the variance estimator. \Cref{thm:estimator-with-optimal-estimated-precesion} (i) shows the asymptotic normality of network dependent estimator with asymptotic variance equal to the optimal adjusted estimator, thus demonstrating its optimality. Since $\sigma_\star^2(\tilde{\bs{\beta}}_\star^\opt) = \argmin_{\bs{\beta}\in \mathbf{R}^Q}\sigma_\star^2(\bs{\beta})\leq \sigma_\star^2(\bs{0})$, \Cref{thm:estimator-with-optimal-estimated-precesion} (ii) shows that under the choice of $\phi_{0}$ and $\hat{\bs{\beta}}_{\star,\ND}$, the resulting estimator is more efficient than without using $\bs{G}$ at all. We emphasize that the procedure $\phi_0$ is essential not only for ensuring optimality but also for guaranteeing the asymptotic normality. Since $\tau_{\phi_0(\bs{G}),i} = \bs{0}$, $(i\in \mathcal{N}_n)$, $n^{1/2}(\hat{\tau}_{\star}(\bs{\beta})-\tau)$ with $\phi\equiv\phi_0$ remains asymptotic normal for any $\bs{\beta}$, providing a class of asymptotic normal estimators to choose from. Fisher’s regression-based estimators using $\phi_0(\bs{G}_i)$ falls within the class of $\hat{\tau}_{\star}(\bs{\beta})$, ensuring its asymptotic normality as well. This means we can include $\phi_0(\bs{G}_i)$ in Fisher's regression, just as we do with $\bs{X}_i$. However, while this approach guarantees asymptotic normality, it does not necessarily lead to an improvement in efficiency compared to no regression adjustment—a characteristic also observed in Fisher's regression with $\bs{X}_i$.

\section{Simulation studies and real data analysis }\label{sec:experiment}

In this section, we conduct numerical simulation studies and analyze a real-world experiment to evaluate the finite-sample performance of our network-dependent estimators with general auxiliary variables.

\subsection{Simulation studies}

To facilitate a fair comparison of our network-dependent estimators with established benchmarks in \citet{gao2023causal} and \citet{leung2022causal}, we adopt a similar data generation process as theirs.

\paragraph{Data generation and experimental setting}

The data generation process inherited from~{\cite{gao2023causal} includes the linear and nonlinear settings which is generated as follows:
\begin{equation*}
    \begin{aligned}
    V_i := \alpha_0 + \alpha_1 \frac{\sum_{j=1}^n A_{i j} Y_j}{\sum_{j=1}^n A_{i j}}+\alpha_2 \frac{\sum_{j=1}^n A_{i j} D_j}{\sum_{j=1}^n A_{i j}}+\alpha_3 D_i+ {\alpha_4} {\bs{X}_i}+\varepsilon_i \text{~with~} Y_i = f(V_i).
    \end{aligned}
\end{equation*}
We choose the population size $n=3000$ and $ \prob(D_i=1) = 0.5$ for each node. We specify (i) for the linear-in-means model, $f(t) = t$ with $(\alpha_0, \alpha_1, \alpha_2, \alpha_3, \alpha_4 )=(-1,0.1,1,1,1)$; (ii) for the nonlinear contagion model, $f(t) = \bs{1}(t\geq 0)$ with $(\alpha_0, \alpha_1, \alpha_2$, $\alpha_3, \alpha_4) = (-1,1.5,1,1,1)$. We relegate more details of the generation process to the Supplementary Material.

For the adjacency matrix $\bs{A}$, we first independently generate the two-dimensional coordinate $\bs{\rho}_i$ uniformly in $[0,1]^2$ for each node $i$. For each pair of nodes $(i,j)$, if $||\bs{\rho}_i - \bs{\rho}_j||_2 \leq \sqrt{1.5/ (\pi n)}$, then we add an edge between them. We delete the edgeless nodes to avoid violating Assumption~\ref{a:overlap}. For the noise $\bs{\varepsilon}$, we choose $\varepsilon_i:= v_i + \rho_{i1} -0.5$ to generate the correlation between node noise, where $\rho_{i1}$ is the first coordinate of $\bs{\rho}_i$ and $v_i \overset{\text{i.i.d.}}{\sim}\mathcal{N}(0,1)$. Moreover, we generate each covariate $\bs{X}_i \overset{\text{i.i.d.}}{\sim}\mathcal{N}(0,1)$.

\paragraph{Repeated sampling evaluation}  After we generate the adjacency matrix $\bs{A}$, noise $\bs{\varepsilon}$, and covariates $\bs{X}$ in the first draw, we keep them fixed throughout the simulation study. Then, we sample $\bs{D}$ for $10^4$ times and evaluate the performance of different methods in the $10^4$ random assignments.

We choose the exposure mapping $T_i = \bs{1}(\sum_{j=1}^n A_{ij}D_j > \lfloor \sum_{j=1}^n A_{ij}/2 \rfloor)$, and then compute the corresponding propensity scores $\pi_i(1), \pi_i(0)$ accordingly. We consider the target $\tau = \mu(T_i = 1) - \mu(T_i = 0)$, i.e., the average treatment effect of more than half of the neighbors treated versus the opposite.

 Moreover, we compare these estimators' Oracle standard errors (Oracle SE), empirical absolute bias and 
the empirical average of the estimated standard errors (Estimated SE) in $10^4$ random assignments. Here, the Oracle SE of $\hat{\tau}$ is defined as $\Var(\hat{\tau})^{\frac{1}{2}}$, approximated via the $10^4$ random assignments. On the other hand, the Estimated SE is derived by the HAC estimator in \Cref{sec:varest}. We choose $b_n = 3$ for $B_{ij}$ in our variance estimator. Finally, we report the empirical coverage probabilities of $95\%$ confidence intervals derived by Oracle SE and Estimated SE denoted by ``Oracle coverage probability" and ``Empirical coverage probability", respectively.

\paragraph{Method comparison} 
We report the performance of estimators $\hat{\tau}_{\haj, \NDF}$ and $\hat{\tau}_{\haj, \NDL}$ defined in~\Cref{sec:ndr}, denoted by \NDF, \NDL. We also report $\hat{\tau}_{\haj, \ND}$ defined in~\Cref{sec:nddnorm}, with the general auxiliary variables $\phi_{0} (\bs{G}_{1,i})$ and $\phi_{0} (\bs{G}_{2,i})$, respectively; where motivated by \Cref{example:auxiliary-variable-linear-in-means}, the $\bs{G}_{1,i}$'s and $\bs{G}_{2,i}$'s are defined as
\begin{equation}
\begin{aligned}
\bs{G}_{1} &:= \{\bs{G}_{1,i}\}_{i \in \mathcal{N}_n} = \big\{\big( D_i,~ \sum_{j=1}^n A_{ij}D_j /\sum_{j=1}^n A_{ij}, ~{\bs{X}}_i,~ \sum_{j=1}^n A_{ij}\bs{X}_j /\sum_{j=1}^n A_{ij} \big)^\top \big\}_{i \in \mathcal{N}_n},\\
\bs{G}_2 &:= \{\bs{G}_{2,i}\}_{i \in \mathcal{N}_n} = \{(\bs{G}^\top_{1,i}\bs{1}(T_i = 1),\bs{G}_{1,i}^\top\bs{1}(T_i =0) )^\top \}_{i \in \mathcal{N}_n}.
\end{aligned}\label{syn_def_G}
\end{equation}
We denote these estimators by $\ND$-$\phi_{0}(G_1)$ and $\ND$-$\phi_{0}(G_2)$. To implement $\phi_{0}(\cdot)$, i.e., the normalizing procedure introduced in~\eqref{eq:point-wise-decorrelation}, we conduct $10^5$ random assignments to approximate $\bs{\gamma}_i$.  We also report the performance of the estimators $\text{ND-}{G}_{\text{1}}$ and $\text{ND-}{G}_{\text{2}}$, which corresponds to the case without using normalization procedure~\eqref{eq:point-wise-decorrelation}. Moreover, we report $\text{F-}\phi_0({G}_\text{1}), \text{F-}\phi_0({G}_\text{2})$, which is constructed by substituting the covariates $\bs{X}_i$ in Fisher's regression with $\phi_0(\bs{G}_{1,i})$ and $\phi_0(\bs{G}_{2,i})$, respectively. 
For benchmarks, we report the HT estimator without covariates (HT) and the three methods in \citet{gao2023causal}: (i) Hajek estimator without covariates (Haj), (ii) Fisher's regression-based estimator with covariates $\bs{X}_i$ (F), (iii) Lin's regression-based estimator with interacted covariates $({\bs{X}^\top_i}\bs{1}(T_i = 1), \bs{X}^\top_i\bs{1}(T_i =0) )^\top$ (L). 
\begin{table}[t]
\scalebox{0.60}{
\begin{tabular}{c cccc cc cc cccc}

\hline\hline Outcome model & \multicolumn{12}{c}{ Linear-in-Means } \\
\hline Method & $\text{HT}$ & $\text{Haj}$ & ${\text{F}}$  & $\text{L}_{\text{}}$ & ${\text{F-}\phi_0({G}_1)}$ & $\text{F-}\phi_0({G}_2)$  & $\text{ND-F}$ &$\text{ND-}\phi_0({G}_1)$ &$\text{ND-}{G}_{\text{1}}$  & $\text{ND-L}_{\text{}}$  &  $\text{ND-}\phi_0({G}_2)$ & $\text{ND-}{G}_2$ \\
\hline
\text{Empirical absolute bias}  &0.001 &0.001 &0.000 &0.000 & 0.001 & 0.001 &\textbf{0.000} &\textbf{0.001} &0.895 &\textbf{0.000} &\textbf{0.002} & 0.100\\
\hline Oracle SE &0.295 &0.293 &0.260 &0.260 &0.235& 0.233&\textbf{0.260} &\textbf{0.235} & 0.466&\textbf{0.260} &\textbf{0.236} & 0.822\\
Estimated SE  &0.294 &0.292 &0.259 &0.259 &0.232&0.232 &\textbf{0.259} &\textbf{0.231} & 0.230 &\textbf{0.259} &\textbf{0.233} &0.230\\
 \hline Oracle coverage probability & 0.949& 0.949 &0.950 &0.950 & 0.953& 0.951 & \textbf{0.953} &\textbf{0.949}&0.015& \textbf{0.950} &\textbf{0.949}&0.679\\
 Empirical coverage probability & 0.946& 0.947 &0.948 &0.948 & 0.944& 0.946 & \textbf{0.947} &\textbf{0.945}&0.000& \textbf{0.947} &\textbf{0.941}&0.039\\

\hline\hline Outcome model & \multicolumn{12}{c}{ Nonlinear Contagion } \\
\hline Method & $\text{HT}$ & $\text{Haj}$ & ${\text{F}}$  & $\text{L}_{\text{}}$ & ${\text{F-}\phi_0({G}_1)}$ & $\text{F-}\phi_0({G}_2)$  & $\text{ND-F}$ &$\text{ND-}\phi_0({G}_1)$ &$\text{ND-}{G}_{\text{1}}$  & $\text{ND-L}_{\text{}}$  &  $\text{ND-}\phi_0({G}_2)$ & $\text{ND-}{G}_2$ \\
\hline
\text{Empirical absolute bias} &0.000 &0.000 &0.000 &0.000 &0.000 & 0.000&\textbf{0.000} &\textbf{0.000} &0.176&\textbf{0.000} &\textbf{0.000}&0.044\\
\hline Oracle SE &0.227 &0.142 &0.136 &0.136 &0.131 &0.129 &\textbf{0.136} &\textbf{0.132} &0.271&\textbf{0.136} &\textbf{0.131}&0.461\\
 Estimated SE &0.228 &0.146 &0.140 &0.139 &0.134 &0.133  &\textbf{0.140} &\textbf{0.134} &0.134&\textbf{0.139} &\textbf{0.133}&0.132\\
\hline 
Oracle coverage probability &0.950 &0.947 &0.949 &0.949 &0.947&0.948  & \textbf{0.947} &\textbf{0.949} &0.337&\textbf{0.949} &\textbf{0.949}&0.945\\
 Empirical coverage probability &0.952 &0.959 &0.961 &0.959 &0.957 & 0.961 & \textbf{0.960} &\textbf{0.952} &0.024&\textbf{0.958} &\textbf{0.958}&0.134\\
\hline\hline
\end{tabular} } \caption{Simulation results: network size $n=3000$ ($b_n = 3$). $\tau = 0.934$ for the linear-in-means setting and $\tau = 0.193$ for the nonlinear contagion setting.} \label{tab_1}
\end{table}

\paragraph{Results} We present the result in Table~\ref{tab_1}. First, for the network-dependent estimators, our $\text{ND-F}, \text{ND-L}$ decrease the Oracle SE by $11.86\%$, compared with $\text{HT}, \text{Haj}$. It validates the efficiency improvement of our network-dependent estimators compared with the unadjusted ones. There is no significant improvement in Oracle SE and Estimated SE when comparing the \text{ND-F} with the F. However, in some instances, the precision of F, along with L and ND-L, is even inferior to that of HT and Haj as demonstrated by the counterexample in the Supplementary Material. In contrast, \text{ND-F} ensures an improvement in efficiency compared to the unadjusted estimators, as we mentioned in~\Cref{sec:review-of-Gao-and-Ding-regression}.


Second, equipped with the general auxiliary variables, $\text{ND-}\phi_0({G}_{\text{1}})$ and $\text{ND-}\phi_0({G}_{\text{2}})$ improve the efficiency, compared with \text{ND-F} and $\text{ND-L}$. For instance, in the linear-in-means setting, the improvement of $\text{ND-}\phi_0({G}_{\text{1}})$, $\text{ND-}\phi_0({G}_{\text{2}})$ upon Oracle SE is at least $9.23\%$ compared with $\text{ND-}\text{F}$, $\text{ND-}\text{L}_\text{}$. It validates that our general auxiliary variables $\bs{G}_{1,i}, \bs{G}_{2,i} $, which incorporate the network information, can gain more efficiency for our network-dependent estimator than the traditional covariates $\bs{X}_i$ and $({\bs{X}^\top_i}\bs{1}(T_i = 1), \bs{X}^\top_i\bs{1}(T_i =0) )^\top$. 

Third, results show that empirical absolute bias of $\text{ND-F}$, $\text{ND-L}$, $\text{ND-}\phi_0{({G}_\text{1})}$, $\text{ND-}\phi_0({G}_\text{2})$ is lower than $ 0.002$, which is negligible. In contrast, the estimators $\text{ND-}{G}_\text{1}$, $\text{ND-}{G}_\text{2}$ which lack the normalizing procedure are inconsistent. It demonstrates that the normalizing procedure is necessary.

Noteworthy, there is no significant difference between the performance of $\text{ND-}\phi_0{({G}_\text{1})}$ and $\text{ND-}\phi_0{({G}_\text{2})}$, or \text{ND-F} and $\text{ND-L}$. It is potentially due to the potential outcomes not exhibiting significant heterogeneity between different exposure mapping values. Finally, the empirical coverage probabilities of all these methods are close to $0.95$.


\subsection{Real data analysis}
In this section, we reanalyze the experiment in~\cite{cai2015social} to compare our network-dependent estimators with the current benchmarks~\citep{leung2022causal,gao2023causal}.

\cite{cai2015social} studied the effect of insurance seminars on Chinese farmers’ purchasing intentions within a social network by a randomized experiment. $D_i \in \{0,1\}$, which will be specified later. $Y_i$ is the subsequent purchasing behavior of farmer $i$, with $Y_i=1$ if the farmer $i$ purchased the insurance after the seminars.

To formally define $D_i$, it will be helpful to introduce the design of the experiment. This experiment was composed of two large-scale insurance promotional lectures. There was a three-day delay between these two lectures. Such a gap allowed farmers who attended the first lecture to share the insurance information they acquired with their friends. In each lecture, there were two parallel sessions, one with an intense introduction to the insurance plan, and one with a simpler introduction. In the experiment, the farmers were first randomly assigned into two subgroups, each participating in different lectures; then in each lecture, the farmers were further randomly assigned into two subgroups to participate in different sessions. We use $\text{Delay}_i, \text{Int}_i \in \{0,1\}$ to represent the lecture and session farmer $i$ was assigned to take, respectively. In specific, $\text{Delay}_i = 1$ represents the delayed lecture, and $\text{Int}_i = 1$ represents the intensive session. $ (\text{Delay}_i, \text{Int}_i) \in \{(0,0), (0,1), (1,0), (1,1)\}$, appear $\{1078, 1098,1353, 1373\}$ times in the original dataset, respectively. Following \cite{cai2015social}, our analysis focuses on investigating the impact of an individual’s and their neighbor’s attendance at the intensive session on the individual’s final purchase choices, i.e., the direct and spillover effect of the intensive session. Since the allocation of $\text{Delay}_i \in \{0,1\}$ is independent of $\text{Int}_i \in \{0,1\}$, we condition on $\{\text{Delay}_i\}_{i \in \mathcal{N}_n}$ and then assume that the randomness in this experiment solely comes from $D_i  \equiv \text{Int}_i$, which conforms i.i.d. to a Bernoulli distribution with probability $0.5$. 

We define two exposure mappings for direct effect and spillover effect following~\cite{gao2023causal} and \cite{leung2022causal}: (i) $T_i = \text{Int}_i$ for direct effect and (ii) $T_i = \bs{1}(\sum_{j=1}^n A'_{ij}(1-\text{Delay}_j) \text{Int}_j>0)$ (whether having at least one friend participating in the first lecture’s intensive session) for spillover effect. Here $\bs{A'} = \{A'_{ij}\}_{i,j \in \mathcal{N}_n}$ is the directed social network obtained by questionnaire, where each farmer filled in the five closest friends of theirs. We set $A'_{ij} = 1$ when $j$ is among the top five friends of $i$, otherwise $A'_{ij} = 0$. We construct the undirected adjacency matrix $\bs{A} = \{A_{ij}\}_{i,j\in \mathcal{N}_n}$, where $A_{ij} =  \max \{A'_{ij}, A'_{ji}\} {}$ for the HAC variance estimator (choosing $b_n = 3$) as in~\Cref{sec:varest}.

We restrict the target population as follows: those (i) who attended the second round session ($\text{Delay}_i = 1$), and (ii) whose friends were not all in either the first or the second lecture, i.e., $\sum_{j=1}^n A'_{ij} (1-\text{Delay}_j) \ne 0$ and $\sum_{j=1}^n A'_{ij} \text{Delay}_j \ne 0$.

We conduct $10^4$ repetitions to compute the normalizing procedure~\eqref{eq:point-wise-decorrelation}. From the dataset, we select seven covariates for each farmer: \{household size, rice area, rice income, education, repayment of insurance, comprehension of the session, village affiliation\} and denote them as $\bs{X}$. 
Inspired by the linear regression in~\cite{cai2015social}, we construct $\bs{G}_{1,i} = (\bs{X}^T_i, \sum_{j}{A}'_{ij}{(1-\text{Delay}_j)\text{Int}_j}/\sum_{j}{A}'_{ij}, \sum_{j}{A}'_{ij}{\text{Int}_j}/\sum_{j}{A}'_{ij})^\top$ in the simulation studies ; and then construct $\bs{G}_{2,i} = (\bs{G}_{1,i}^\top\bs{1}(T_i=1),\bs{G}_{1,i}^\top\bs{1}(T_i=0))^\top$. Here, the second and third terms of $\bs{G}_{1,i}$ describe the proportions of the neighborhood nodes assigned to the intensive session in the first lecture, and to the intensive sessions, respectively.

\begin{table}[t]
\scalebox{0.9}{
\begin{tabular}{c cccc cccc}

\hline\hline \multicolumn{9}{c}{Direct effect ($T_i = \text{Int}_i$)} \\
\hline Method & $\text{HT}$ & $\text{Haj}$ & ${\text{F}}$ & $\text{L}_{\text{}}$   & $\text{ND-F}$ &$\text{ND-}\phi_0({G}_1)$ & $\text{ND-L}$  &  $\text{ND-}\phi_0({G}_2)$ \\
\hline $\hat{\tau}$ & 0.0082 & 0.0146 &  0.0170 &0.0168 & \textbf{0.0167 }& \textbf{0.0161} &\textbf{0.0164}&\textbf{0.0175}\\
\hline Estimated SE & 0.0272 & 0.0225 & 0.0218 & 0.0197 & \textbf{0.0211} & \textbf{0.0205} & \textbf{0.0195} &    \textbf{0.0193} \\
\hline\hline \multicolumn{9}{c}{ Spillover effect ($T_i = \sum_{j=1}^n A'_{ij} (1-\text{Delay}_j) \text{Int}_j > 0$ )} \\
\hline Method & $\text{HT}$ & $\text{Haj}$ & ${\text{F}}$ & $\text{L}_{\text{}}$   & $\text{ND-F}$ &$\text{ND-}\phi_0({G}_1)$ & $\text{ND-L}$  &  $\text{ND-}\phi_0({G}_2)$ \\
\hline
$\hat{\tau}$&  0.0381 & 0.0611 & 0.0604 & 0.0581  & \textbf{0.0660} & \textbf{0.0686} & \textbf{0.0592} & \textbf{0.0603}  \\
\hline Estimated SE & 0.0447 & 0.0292 & 0.0270 & 0.0241 & \textbf{0.0258} & \textbf{0.0250} & \textbf{0.0236} & \textbf{0.0233}\\
\hline\hline
\end{tabular} } \caption{Estimates $\hat{\tau}$ and Estimated SE of empirical experiments.} \label{tab_empirical}
\end{table}

We consider the direct effect and spillover effect, respectively. We present $\text{HT}$, $\text{Haj}$, $\text{F}$, $\text{L}$, $\text{ND-F}$, $\text{ND-}\phi_0({G}_1)$, $\text{ND-L}$, $\text{ND-}\phi_0({G}_2)$, which are all inherited from the simulation studies. Given a single randomized experiment in~\citet{cai2015social}, we report the {point estimate} $\hat{\tau}$ and {estimated standard error (Estimated SE)}.

The results are presented in Table~\ref{tab_empirical} and Figure~\ref{fig_empirical}. For the direct effect, the point estimates obtained by these methods are all positive. However, as shown in Figure~\ref{fig_empirical}, they are all insignificant. This suggests that we do not have enough evidence to support the claim that participating in intensive sessions affects farmers' purchasing intention, which aligns with~\cite{cai2015social}. Nevertheless, our $\text{ND-}\phi_0({G}_{\text{1}})$ reduces the Estimated SE by $5.96\%$ compared with $\text{F}$. 

For the spillover effect, except for \text{HT}, the $95\%$ confidence intervals of all the regression-based methods are above $0$; this is also in line with Section 4 of~\cite{cai2015social}, where they claim the spillover effect in this experiment is significantly positive. Similarly, the Estimated SE of our $\text{ND-}\phi_0({G}_{\text{1}})$ is $ 7.41\%$ smaller than that of $\text{F}$~\citep{gao2023causal}.

In sum, regression-based estimators ($\text{F}, \text{L}_{}$ and our network-dependent estimators $\text{ND-F}$, $\text{ND-}\phi_0({G}_1)$, $\text{ND-L}$, $\text{ND-}\phi_0({G}_2)$) all exhibit a smaller Estimated SE than unadjusted HT~\citep{leung2022causal} and Haj~\citep{gao2023causal}; 
and ND-$\phi_0({G}_2)$ and ND-L outperform all other methods in terms of Estimated SE with ND-$\phi_0({G}_2)$ slightly exceeding ND-L.

\begin{figure}[t]
    \centering
\includegraphics[width=0.9\textwidth]{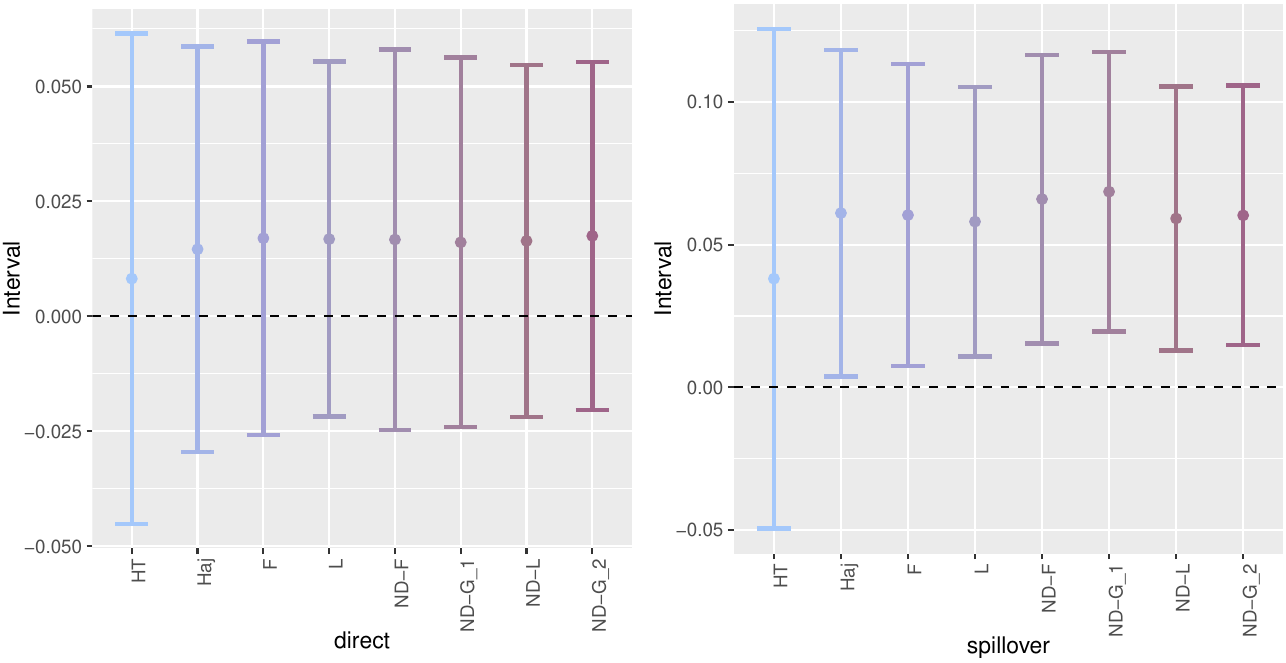}
    \caption{$95\%$ confidence interval of direct effect (left)/ spillover effect (right) via different methods.
    }
    \label{fig_empirical}
\end{figure}

\section{Conclusion}
\label{sec:conclusion}
In this article, we consider regression adjustment with experimental data in the presence of network interference. Unlike the SUTVA setting, we leverage the network information by defining the auxiliary variables to improve the estimation efficiency. Moreover, we obtain estimators that never hurt the asymptotic efficiency compared to the unadjusted estimators by using a network-dependent regression adjustment and a special normalizing function. We use the new regression adjustment method to replicate the results from \cite{cai2015social}, estimating the direct and spillover effects of insurance seminars on Chinese farmers' 
purchasing intentions. Our methods yield shorter confidence intervals compared to the existing methods.

We focus on the regime of a fixed number of covariates. It would be interesting to extend the framework for a diverging number of covariates \citep{lei2021regression,lu2023debiased}.
Moreover, rerandomization is frequently used in the design stage to balance covariates \citep{morgan2012rerandomization,li2018asymptotic,wang2022rerandomization}. With the presence of network interference, we conjecture that balancing the auxiliary variables that incorporate network information can lead to a further precision improvement than the classic rerandomization that leverages only individual-level characteristics.

\bibliographystyle{chicago}
\bibliography{causal}

\newpage

\appendix

\centerline{ \Large\bf SUPPLEMENTARY MATERIAL}
\vspace{2mm}

\spacingset{1.5}
\noindent Section~\ref{sec:A} provides the proofs of Proposition~\ref{prop:bias-of-variance-estimator-Lin-Fisher} and \Cref{prop:bias-of-variance-estimator}

\noindent Section~\ref{sec:B} provides proofs of \Cref{thm:estimator-with-optimal-estimated-precesion-Lin-Fisher} and \Cref{thm:estimator-with-optimal-estimated-precesion}.

\noindent Section~\ref{sec:counter} presents a constructive setting and an additional simulation study when the asymptotic variance of Fisher's and Lin's regression adjusted estimator is strictly less efficient than the unadjusted estimator.

\noindent Section~\ref{sec:details-gen} provides additional details of the data-generation process in the simulation study.


\section{Proofs of Proposition~\ref{prop:bias-of-variance-estimator-Lin-Fisher} and \Cref{prop:bias-of-variance-estimator}}
\label{sec:A}
\begin{lemma}
\label{lem:order-of-haj-ht-mean}
 Under Assumptions \ref{a:overlap}-\ref{a:weak-dependency-for-LLN}, \ref{a:bounded-G-dc}, we have
\begin{align*}
    &\hat{\mu}_{\star}(\bs{t})-\mu(\bs{t}) =\Op\Big( n^{-1/2}(1+\sum_{s=0}^n M_n^\partial (s)\tilde{\theta}_{n,s})^{1/2}\Big),\\
    &\hat{\bs{\mu}}_{\star, \bs{G}}(\bs{t})-\bs{\mu}_{\bs{G}}(\bs{t}) = \Op\Big(n^{-1/2}(1+\sum_{s=0}^n M_n^\partial (s)\tilde{\theta}_{n,s})^{1/2}\Big).
\end{align*}  
\end{lemma}
\begin{proof}[Proof of \Cref{lem:order-of-haj-ht-mean}]
Let $\bs{G}_i = (\bs{G}_{iq})_{q=1}^Q$, $\bs{\mu}_{\bs{G}}(\bs{t}) = (\mu_{{G}_q}(\bs{t}))_{q=1}^Q$ and $\hat{\bs{\mu}}_{\star, \bs{G}}(\bs{t}) = (\hat{{\mu}}_{\star,{G}_q}(\bs{t}))_{q=1}^Q$, $\star\in\{\HT,\haj\}$.
    Define $1_{\HT}(\bs{t})=n^{-1}\sumi  \bs{1}(\bs{T}_i = \bs{t})/\pi_i(\bs{t})$.
    By the proof of \citet[Theorem 2]{leung2022causal} with $Y_i\equiv Y_i\bs{1}(\bs{T}_i=\bs{t}),{G}_{iq}\bs{1}(\bs{T}_i=\bs{t})$, and $\bs{1}(\bs{T}_i=\bs{t})$, we have from Assumptions \ref{a:overlap}-\ref{a:weak-dependency-for-LLN}, \ref{a:bounded-G-dc} that
    \begin{align*}
        &\hat{\mu}_{\HT}(\bs{t})-\mu(\bs{t}) = \Op\Big( n^{-1/2}(1+\sum_{s=0}^n M_n^\partial (s)\tilde{\theta}_{n,s})^{1/2}\Big),\\ &\hat{\mu}_{\HT,G_q}(\bs{t})-\mu_{G_q}(\bs{t}) = \Op\Big( n^{-1/2}(1+\sum_{s=0}^n M_n^\partial (s)\tilde{\theta}_{n,s})^{1/2}\Big),\\
        &1_{\HT}(\bs{t})-1 = \Op\Big( n^{-1/2}(1+\sum_{s=0}^n M_n^\partial (s)\tilde{\theta}_{n,s})^{1/2}\Big).
    \end{align*}
    Assumption~\ref{a:weak-dependency-for-LLN} implies that $1_{\HT}(\bs{t}) = 1+\op(1)$ and thereby
    \[
    1/1_{\HT}(\bs{t}) = 1+ \op(1).
    \]
     As a consequence, we have
    \begin{align*}
        &\hat{\mu}_{\haj}(\bs{t}) = \hat{\mu}_{\HT}(\bs{t})/1_{\HT}(\bs{t}) = \mu(\bs{t}) + \Op\Big( n^{-1/2}(1+\sum_{s=0}^n M_n^\partial (s)\tilde{\theta}_{n,s})^{1/2}\Big)\\
        &\hat{\mu}_{\haj,G_q}(\bs{t})  = \hat{\mu}_{\HT,G_q}(\bs{t})/1_{\HT}(\bs{t}) = \mu_{G_q}(\bs{t}) + \Op\Big( n^{-1/2}(1+\sum_{s=0}^n M_n^\partial (s)\tilde{\theta}_{n,s})^{1/2}\Big).
    \end{align*}
   
    The conclusion follows.
\end{proof}

We first prove \Cref{prop:bias-of-variance-estimator}. To do this, we first prove that $Y_i-\phi(\bs{G}_i)^\top \bs{\beta}$ satisfies Assumptions \ref{a:bounded-outcome}, \ref{a:ANI}, \ref{a:weak-dependency-for-LLN} for $Y_i$ and then apply Theorem~4 from \cite{leung2022causal} with $Y_i\equiv Y_i-\phi(\bs{G}_i)^\top \bs{\beta}$. For ease of notation, we write $\tilde{\bs{G}}_i\equiv \phi(\bs{G}_i)$. Since \Cref{a:bounded-G-dc} holds for $\phi(\bs{G}_i)$, we have, for $i\in\mathcal{N}_n$ and $d\in \{0,1\}^n$,
\begin{align*}
   |Y_i(\bs{d})-\tilde{\bs{G}}_i(\bs{d})^\top \bs{\beta}| \leq |Y_i(\bs{d})|+\|\tilde{\bs{G}}_i(\bs{d})\|_{\infty} \|\bs{\beta}\|_{1} \leq  c_{Y}+\|\bs{\beta}\|_1c_{\bs{G}}
\end{align*}
 and
\begin{align*}
&\max_{i\in\mathcal{N}_n}\mathbb{E}[|Y_i(\bs{D})-Y_i(\bs{D}^{(i,s)})-\{\tilde{\bs{G}}_i(\bs{D})-\tilde{\bs{G}}_i(\bs{D}^{(i,s)})\}^\top \bs{\beta}|] \\
\leq &\max_{i\in\mathcal{N}_n}\mathbb{E}[|Y_i(\bs{D})-Y_i(\bs{D}^{(i,s)})|] +  \max_{i\in\mathcal{N}_n}\mathbb{E}[\|\tilde{\bs{G}}_i(\bs{D})-\tilde{\bs{G}}_i(\bs{D}^{(i,s)})\|_{\infty}]  \|\bs{\beta}\|_{1} = O(\theta_{n,s}),
\end{align*}
which concludes that $Y_i-\phi(\bs{G}_i)^\top \bs{\beta}$ satisfies Assumption \ref{a:bounded-outcome}, \ref{a:ANI}, \ref{a:weak-dependency-for-LLN} for $Y_i$. 

Therefore, by \cite[Theorem~4]{leung2022causal} with $Y_i\equiv Y_i-\phi(\bs{G}_i)^\top \bs{\beta}$, we have
\[
\hat{\sigma}^2_{\HT}(\bs{\beta})={\sigma}^2_{\HT}(\bs{\beta})+R(\bs{\beta})+\op(1).
\]

For $\hat{\sigma}^2_{\haj}(\bs{\beta})$, define $\hat{\bs{V}}_{\star,\tilde{\bs{G}},i}$ and $\bs{V}_{\star,\tilde{\bs{G}},i}$ analogously as $\hat{V}_{\star,i}$ and $V_{\star,i}$, with $Y_i$ replaced with $\tilde{\bs{G}}_i$. $\hat{V}_{\haj,i}(\bs{\beta})$ and $V_{\haj,i}(\bs{\beta})$ can be expressed as $\hat{V}_{\haj,i}(\bs{\beta}) = \hat{V}_{\haj,i}-\hat{\bs{V}}_{\haj,\tilde{\bs{G}},i}^\top\bs{\beta}$ and $V_{\haj,i}(\bs{\beta}) = {V}_{\haj,i}-{\bs{V}}_{\haj,\tilde{\bs{G}},i}^\top\bs{\beta}$, respectively. We have $\E V_{\haj,i}(\bs{\beta}) = \tau_i-\tau -\bs{\beta}^\top(\bs{\tau}_{\tilde{\bs{G}},i}-\bs{\tau}_{\tilde{\bs{G}}})$ and $n^{-1}\sumi  \sumi  \E V_{\haj,i}(\bs{\beta})\E V_{\haj,j}(\bs{\beta}) B_{ij} = R(\bs{\beta})$. We have
\begin{align*}
    &\hat{\sigma}^2_{\haj}(\bs{\beta}) = n^{-1}\sumi  \sumj \hat{V}_{\haj,i}(\bs{\beta}) \hat{V}_{\haj,j}(\bs{\beta}) B_{ij} \\
    =& n^{-1}\sumi  \sumj \{V_{\haj,i}(\bs{\beta})-\E V_{\haj,i}(\bs{\beta})\}\{V_{\haj,j}(\bs{\beta})-\E V_{\haj,j}(\bs{\beta})\} B_{ij}+ R(\bs{\beta})+\\
    & n^{-1}\sumi  \sumj  \{\hat{V}_{\haj,i}(\bs{\beta})\hat{V}_{\haj,j}(\bs{\beta})-V_{\haj,i}(\bs{\beta}) V_{\haj,j}(\bs{\beta})\}B_{ij} + \\
    & 2n^{-1}\sumi  \sumj  \{V_{\haj,i}(\bs{\beta})-\E V_{\haj,i}(\bs{\beta})\} \E V_{\haj,j}(\bs{\beta}) B_{ij}. 
\end{align*}

Since $Y_i-\phi(\bs{G}_i)^\top \bs{\beta}$ satisfies Assumption \ref{a:bounded-outcome}, \ref{a:ANI}, \ref{a:weak-dependency-for-LLN} for $Y_i$ and by the proof of \cite[Theorem~4]{leung2022causal}, we have
\begin{align*}
    &n^{-1}\sumi  \sumi \{V_{\haj,i}(\bs{\beta})-\E V_{\haj,i}(\bs{\beta})\}\{V_{\haj,j}(\bs{\beta})-\E V_{\haj,j}(\bs{\beta})\} B_{ij} = \sigma_{\haj}^2(\bs{\beta}) +\op(1),\\
    &n^{-1}\sumi  \sumi  \{V_{\haj,i}(\bs{\beta})-\E V_{\haj,i}(\bs{\beta})\} \E V_{\haj,j}(\bs{\beta}) B_{ij} = \op(1).
\end{align*}

It remains to prove 
\[
n^{-1}\sumi  \sumi  \{\hat{V}_{\haj,i}(\bs{\beta})\hat{V}_{\haj,j}(\bs{\beta})-V_{\haj,i}(\bs{\beta}) V_{\haj,j}(\bs{\beta})\}B_{ij} =\op(1).
\]

We see that
\begin{align*}
    \hat{V}_{\haj,i}(\bs{\beta})-V_{\haj,i}(\bs{\beta}) = \big[\mu(\bs{t})-\hat{\mu}_{\haj}(\bs{t})-\bs{\beta}^\top\{\bs{\mu}_{\tilde{\bs{G}}}(\bs{t})-\hat{\bs{\mu}}_{\tilde{\bs{G}},\haj}(\bs{t})\}\big]\bs{1}(\bs{T}_i=\bs{t})/\pi(\bs{t})-\\
    \big[\mu(\bs{t}^\prime)-\hat{\mu}_{\haj}(\bs{t}^\prime)-\bs{\beta}^\top\{\bs{\mu}_{\tilde{\bs{G}}}(\bs{t}^\prime)-\hat{\bs{\mu}}_{\tilde{\bs{G}},\haj}(\bs{t}^\prime)\}\big]\bs{1}(\bs{T}_i=\bs{t}^\prime)/\pi(\bs{t}^\prime).
\end{align*}
Assumption \ref{a:ANI} and \ref{a:consistency-for-variance-estimator} (i) implies that $\sum_{s=0}^n M_n^\partial (s)\tilde{\theta}_{n,s} = O(1)$ and therefore by \Cref{lem:order-of-haj-ht-mean}, we have
\begin{align*}
    &\max_{i\in \mathcal{N}_n}|\hat{V}_{\haj,i}(\bs{\beta})-V_{\haj,i}(\bs{\beta})| \leq \frac{1}{\underline{\pi}} \big|\mu(\bs{t}^\prime)-\hat{\mu}_{\haj}(\bs{t}^\prime)-\bs{\beta}^\top\{\bs{\mu}_{\tilde{\bs{G}}}(\bs{t}^\prime)-\hat{\bs{\mu}}_{\haj, \tilde{\bs{G}}}(\bs{t}^\prime)\}\big| +\\
    &\qquad \qquad\qquad\qquad\frac{1}{\underline{\pi}} \big|\mu(\bs{t})-\hat{\mu}_{\haj}(\bs{t})-\bs{\beta}^\top\{\bs{\mu}_{\tilde{\bs{G}}}(\bs{t})-\hat{\bs{\mu}}_{\haj, \tilde{\bs{G}}}(\bs{t})\}\big| = \Op(n^{-1/2}).
\end{align*}
 On the other hand, we have by \Cref{a:overlap}, \ref{a:bounded-outcome},  and \ref{a:bounded-G-dc} we have, for $\tilde{\bs{t}}\in \mt$, $|\hat{\mu}_{\haj}(\tilde{\bs{t}})| \leq c_{Y}$, $|\mu(\tilde{\bs{t}})| \leq c_{Y}$, $|\hat{\bs{\mu}}_{\haj,\tilde{\bs{G}}}(\tilde{\bs{t}})| \leq c_{\bs{G}}$, $|\bs{\mu}_{\tilde{\bs{G}}}(\tilde{\bs{t}})| \leq c_{\bs{G}}$ and therefore $\max_{i\in \mathcal{N}_n}\{|\hat{V}_{\haj,i}(\bs{\beta})|,|V_{\haj,i}(\bs{\beta})|\}\leq 2\underline{\pi}^{-1}(2c_{Y}+2c_{\bs{G}}\|\bs{\beta}\|_1) = O(1)$.

As a consequence, we have
\begin{align*}
    &|n^{-1}\sumi  \sumj  (\hat{V}_{\haj,i}(\bs{\beta})\hat{V}_{\haj,j}(\bs{\beta})-V_{\haj,i}(\bs{\beta}) V_{\haj,j}(\bs{\beta}))B_{ij}| \\
    \leq& |n^{-1}\sumi  \sumj  (\hat{V}_{\haj,i}(\bs{\beta})-V_{\haj,i}(\bs{\beta}) )V_{\haj,j}(\bs{\beta}) B_{ij}| + \\
    &\quad |n^{-1}\sumi  \sumi  (\hat{V}_{\haj,j}(\bs{\beta})-V_{\haj,j}(\bs{\beta}) )\hat{V}_{\haj,i}(\bs{\beta}) B_{ij}| \\
    &= O(1)O(n^{-1/2}) n^{-1}\sumi  \sumi   B_{ij} = O(1)\Op(n^{-1/2})M_n\left(b_n, 1\right) = \op(1).
\end{align*}
Hence, we prove \Cref{prop:bias-of-variance-estimator}. For \Cref{prop:bias-of-variance-estimator-Lin-Fisher}, we let $\phi(\bs{G}_i) \equiv (\bs{X}_i^\top\bs{1}(\bs{T}_i=\bs{t}),\bs{X}_i^\top\bs{1}(\bs{T}_i=\bs{t}^\prime))^\top$. To conclude \Cref{prop:bias-of-variance-estimator-Lin-Fisher}, it remains to check $\phi(\bs{G}_i)$ satisfies \Cref{a:bounded-G-dc}. First, we see that for any $i \in \mathcal{N}_n$, by \Cref{a:bounded-X}
\[
\|\phi(\bs{G}_i)\|_{\infty} \leq c_{\bs{X}}/\underline{\pi}.
\]
On the other hand, when $s>2\max\{1,K\}$, we have $\bs{T}_i(\bs{D}^{(i,s)}) = \bs{T}_i(\bs{D})$. As a consequence, we have for any $\tilde{\bs{t}}\in \mt$
\[
\max_{i\in\mathcal{N}_n}\mathbb{E}\Big\|\-\bs{X}_i\{\bs{1}(\bs{T}_i(\bs{D})=\tilde{\bs{t}})-\bs{1}(\bs{T}_i(\bs{D}^{(i,s)})=\tilde{\bs{t}})\}\Big\|_{\infty}=0,\quad \text{for},\quad s>2\max\{1,K\}.
\]

In light of above, we show that $\phi(\bs{G}_i) \equiv (\bs{X}_i^\top\bs{1}(\bs{T}_i=\bs{t}),\bs{X}_i^\top\bs{1}(\bs{T}_i=\bs{t}^\prime))^\top$ satisfies \Cref{a:bounded-G-dc}. Therefore, by \Cref{prop:bias-of-variance-estimator}, we prove \Cref{prop:bias-of-variance-estimator-Lin-Fisher}. 

\section{Proofs of \Cref{thm:estimator-with-optimal-estimated-precesion-Lin-Fisher} and \Cref{thm:estimator-with-optimal-estimated-precesion}}
\label{sec:B}

We prove \Cref{thm:estimator-with-optimal-estimated-precesion-SM}, which is a more general result regarding $\phi(\bs{G}_i)$. Then the results of \Cref{thm:estimator-with-optimal-estimated-precesion-Lin-Fisher} and \Cref{thm:estimator-with-optimal-estimated-precesion} easily follows by applying \Cref{thm:estimator-with-optimal-estimated-precesion-SM} with $\phi(\bs{G}_i) \equiv (\bs{X}_i^\top\bs{1}(\bs{T}_i = \bs{t}), \bs{X}_i^\top\bs{1}(\bs{T}_i = \bs{t}^\prime))^\top$ and $\phi(\bs{G}_i) \equiv \phi_0(\bs{G}_i)$, respectively. Let $\tilde{\bs{\beta}}_{\star, \ND} = \argmin_{\bs{\beta}} {\sigma}^2_{\star}(\bs{\beta})+R(\bs{\beta})$.

\begin{assumption} 
\label{a:for-CLT-of-estimated-optimal-SM}
For $\star\in\{\HT,\haj\}$, We have
(i) Limit inferior of the smallest eigenvalues of the Hessian matrix of  ${\sigma}^2_{\star}(\bs{\beta})+R(\bs{\beta})$ (with respect to $\bs{\beta}$) at $\tilde{\bs{\beta}}_{\star, \ND}$ is greater than zero (ii) $\liminf_{n\rightarrow\infty}\sigma_\star^2(\tilde{\bs{\beta}}_{\star, \ND})>0$
(iii) There exist $\epsilon>0$ and a sequence of positive constants $\left\{m_n\right\}_{n \in \mathbb{N}}$ such that $m_n \rightarrow \infty$, we have
$$
\max \left\{ \frac{1}{n^2} \sum_{s=0}^n\left|\mathcal{H}_n\left(s, m_n\right)\right| \tilde{\theta}_{n, s}^{1-\epsilon},  n^{-1 / 2} M_n\left(m_n, 2\right),  n^{3 / 2} \tilde{\theta}_{n, m_n}^{1-\epsilon}\right\} \rightarrow 0.
$$
\end{assumption}

\begin{assumption}
    \label{a:bias-same-order-with-estimator-SM}
    The following $3$ terms are of order $O(1)$:
    \begin{align*}
        &\frac{1}{n}\sumi\sumj B_{ij} (\tau_i-\tau)(\tau_j-\tau),\quad \frac{1}{n}\sumi\sumj B_{ij} (\tau_j-\tau )(\bs{\tau}_{\phi(\bs{G}),j}-\bs{\tau}_{\phi(\bs{G})}),\\
        &\frac{1}{n}\sumi\sumj B_{ij} (\bs{\tau}_{\phi(\bs{G}),i}-\bs{\tau}_{\phi(\bs{G})})(\bs{\tau}_{\phi(\bs{G}),j}-\bs{\tau}_{\phi(\bs{G})})^\top.
    \end{align*}
\end{assumption}

\begin{theorem}
\label{thm:estimator-with-optimal-estimated-precesion-SM}
Suppose we are under \Cref{a:bounded-G-dc}, but with $\bs{G}_i$ replaced with $\phi(\bs{G}_i)$. Suppose further Assumption~\ref{a:overlap}--\ref{a:ANI}, \ref{a:consistency-for-variance-estimator},  \Cref{a:for-CLT-of-estimated-optimal-SM}--\ref{a:bias-same-order-with-estimator-SM}, with $\bs{\tau}_{\phi(\bs{G})}=0$,  We have 
(i) $n^{1/2}(\hat{\tau}_{\star,\ND}-\tau)/\sigma_\star(\tilde{\bs{\beta}}_{\star, \ND})\xrightarrow{\textnormal{d}} \mathcal{N}(0,1)$ and (ii) $\hat{\sigma}_\star^2(\hat{\beta}_{\star,\ND}) - \sigma_\star^2(\tilde{\bs{\beta}}_{\star, \ND})- R(\tilde{\bs{\beta}}_{\star, \ND}) = \op(1)$. 
\end{theorem}

We define $\psi$-dependence in line with Definition 2.2 of \cite{kojevnikov2021limit} For $d\in \mathbb{N} $, let $\mathcal{L}_d$ be the set of real-valued bounded Lipschitz functions on $\mathbb{R}^d{:}$
$$
\mathcal{L}_{d}:=\{f:\mathbb{R}^{d}\rightarrow\mathbb{R}:\:\|f\|_{\infty}<\infty,\:\mathrm{Lip}(f)<\infty\},
$$

where $\|f\|_{\infty}:=\sup_{\bs{x}\in\mathbb{R}^d}|f(\bs{x})|$ and Lip$(f)$ indicates the Lipschitz constant of $f$, that is $|f(\bs{x}_1)-f(\bs{x}_2)|\leq \operatorname{Lip}(f)\|\bs{x}_1-\bs{x}_2\|_1$. We write the distance between subsets $H,H^{\prime}\subset \mathcal{N}_n$ by $\ell_{\boldsymbol{A}}(H,H^{\prime}):=$ $\min\{\ell_{\bs{A}}(i,j):i\in H,j\in H^{\prime}\}.$ For $h, h^{\prime}\in \mathbb{N},$ denote the collection of pairs $( H, H^{\prime}) $ whose sizes are $h$ and $h^{\prime}$, respectively, with distance at least $s$ as
$$
\mathcal{P}_{n}(h,h^{\prime},s):=\{(H,H^{\prime}):H,H^{\prime}\subset \mathcal{N}_{n},|H|=h,|H^{\prime}|=h^{\prime},\ell_{\bs{A}}(H,H^{\prime})\geqslant s\}.
$$

For a generic random vector $\bs{W}_{n,i}\in\mathbb{R}^v$, let $\bs{W}_{n, H}= ( \bs{W}_{n, i}) _{i\in H}$ and $\bs{W}_{n, H^{\prime}}= ( \bs{W}_{n, i}) _{i\in H^{\prime}}.$
\begin{definition}[$\psi$ Dependent]
     A triangular array $\{\bs{W}_{n,i}\}_{i\in S_n}$ is called $\psi$-dependent, if for each
 $n\in\mathbb{N}$, there exist a sequence of uniformly bounded constants $\{\tilde{\theta}_{n,s}\}_{s\geqslant0}$ with $\tilde{\theta}_{n,0}=1$ and a collection of nonrandom functions $\{\psi_{h,h^{\prime}}\}_{h,h^{\prime}\in\mathbb{N}}$, where $\psi_{h,h^{\prime}}:\mathcal{L}_{hv}\times\mathcal{L}_{h^{\prime}v}\to[0,\infty)$, such that for all $s> 0, ( H, H^\prime) \in \mathcal{P} _{n}( h, h^{\prime}, s) , f\in \mathcal{L} _{hv}, $ and $f^{\prime}\in \mathcal{L} _{h^{\prime}v}, $
$$
|\operatorname{Cov}[f(\boldsymbol{W}_{n,H}),f'(\boldsymbol{W}_{n,H'})]|\leqslant\psi_{h,h'}(f,f')\tilde{\theta}_{n,s}.
$$
The sequence $\{\tilde{\theta}_{n,s}\}_{s\geqslant0}$ is called the \emph{dependence coefficients} of $\{\bs{W}_{n,i}\}_{i\in \mathcal{N}_n}$.
\end{definition}
Define $\bs{V}_{\star,\phi(\bs{G}),i}$ and $\hat{\bs{V}}_{\star,\phi(\bs{G}),i}$ the same as $V_{\star,i}$ and $\hat{V}_{\star,i}$, respectively, with $Y_i$ replaced with $\phi(\bs{G}_i)$. Let $\bs{W}_{i} := (V_{\star,i},\bs{V}_{\star,\phi(\bs{G}),i}^\top)^\top\in\mathbb{R}^{1+Q}$. We have
\begin{lemma}
\label{lem:W-i-psi-dependent}
  Under Assumptions \ref{a:overlap}--\ref{a:ANI}, the triangular array $\{\bs{W}_i\}_{i\in \mathcal{N}_n}$ is $\psi$-weakly dependent with the dependence coefficients $\{\tilde{\theta}_{n,s}\}_{s\geqslant0}$ defined in  \Cref{a:ANI} and \ref{a:weak-dependency-for-LLN}, and
 $\psi_{h,h^{\prime}}(f,f^{\prime})=C[\|f\|_{\infty}\|f^{\prime}\|_{\infty}+h(Q+1)\|f^{\prime}\|_{\infty}\operatorname{Lip}(f)+h^{\prime}(Q+1)\|f\|_{\infty}\operatorname{Lip}(f^{\prime})]$, $\forall h,h^{\prime}\in\mathbb{N},f\in\mathcal{L}_{(1+Q)h},f^{\prime}\in\mathcal{L}_{(1+Q)h^{\prime}}$,
 with some positive constant $C$.
\end{lemma}
\begin{proof}[Proof of \Cref{lem:W-i-psi-dependent}]
    The proof is very similar to that of \cite[Theorem 1]{leung2022causal}. Therefore, we omit it.
\end{proof}

\begin{proof}[Proof of \Cref{thm:estimator-with-optimal-estimated-precesion-SM}]
    
\end{proof}
We first prove that $\hat{\bs{\beta}}_{\star, \ND} = \tilde{\bs{\beta}}_{\star, \ND}+\op(1)$. 
We see that
\[
\hat{\bs{\beta}}_{\star, \ND} = \Big(n^{-1}\sumi\sumj B_{ij} \hat{\bs{V}}_{\star,\phi(\bs{G}),i} \hat{\bs{V}}_{\star,\phi(\bs{G}),j}^\top\Big)^{-1}\Big(n^{-1}\sumi\sumj B_{ij} \hat{\bs{V}}_{\star,\phi(\bs{G}),i} \hat{V}_{\star,j}\Big).
\]

Let $\tilde{\bs{G}}_i\equiv\phi(\bs{G}_i) = (\tilde{\bs{G}}_{iq})_{q=1}^Q$. Let $\hat{\bs{V}}_{\star,\phi(\bs{G}),i} = (\hat{{V}}_{\star,\tilde{G}_q,i})_{q=1}^Q$ and ${\bs{V}}_{\star,\phi(\bs{G}),i} = ({{V}}_{\star,\tilde{G}_q,i})_{q=1}^Q$ and $\bs{\tau}_{\phi(\bs{G}),i} = (\tau_{\tilde{G}_q,i})_{q=1}^Q$. 

For simplicity, we write $\hat{{\bs{W}}}_i =(\hat{V}_{\star,i},\hat{\bs{V}}_{\star,\phi(\bs{G}),i}^\top)^\top\in\mathbb{R}^{1+Q}$ and ${\bs{W}}_i =(V_{\star,i},\bs{V}_{\star,\phi(\bs{G}),i}^\top)^\top\in\mathbb{R}^{1+Q}$. Let $\hat{\bs{W}}_i = (\hat{W}_{iq})_{q=1}^{Q+1}$, ${\bs{W}}_i = ({W}_{iq})_{q=1}^{Q+1}$, $\bar{\bs{W}} = n^{-1}\sum_{i=1}^n \bs{W}_{i}$. We now prove that, for any $1\leq q_1,q_2 \leq Q+1$
\begin{align}
\label{eq:oracle-variance-consistent}
    &n^{-1}\sumi\sumj ({W}_{iq_1}-\E {W}_{iq_1}) ({W}_{jq_2}-\E {W}_{jq_2}) B_{ij} = \Cov(n^{-1/2}\sumi {W}_{iq_1}, n^{-1/2}\sumi {W}_{iq_2}) + \op(1),\\
    &n^{-1}\sumi\sumj \hat{W}_{iq_1}\hat{W}_{iq_2}B_{ij} = n^{-1}\sumi\sumj ({W}_{iq_1}-\E {W}_{iq_1}) ({W}_{jq_2}-\E {W}_{jq_2}) B_{ij} + \nonumber\\
    \label{eq:formula-of-bias-of-variance-estimator}
    &\qquad\qquad\qquad\qquad\qquad n^{-1}\sumi\sumj \E {W}_{iq_1} \E {W}_{jq_2} B_{ij} + \op(1).
\end{align}
  Assumption~\ref{a:ANI}, \ref{a:consistency-for-variance-estimator} implies Assumption 4.1 of \cite{kojevnikov2021limit}. By \Cref{lem:W-i-psi-dependent}, under Assumption~\ref{a:overlap}--\ref{a:ANI} and \ref{a:consistency-for-variance-estimator}, we apply \cite[Proposition 4.1]{kojevnikov2021limit}, thereby leading to 
\[
\Big\|n^{-1}\sumi\sumj (\bs{W}_{i}-\E \bs{W}_{i}) (\bs{W}_{j}-\E \bs{W}_{j}) B_{ij}-\Cov(n^{-1/2}\sumi \bs{W}_{i}, n^{-1/2}\sumi \bs{W}_{i})\Big\|_{\textnormal{F}}= \op(1),
\]
where $\|\cdot\|_{\textnormal{F}}$ is the Frobenius norm. Hence, \eqref{eq:oracle-variance-consistent} holds.

To prove \eqref{eq:formula-of-bias-of-variance-estimator}, we note that the alleged $\op(1)$ term is equal to
\begin{align*}
    n^{-1}\sumi  \sumi  (\hat{W}_{iq_1}\hat{W}_{jq_2}-W_{iq_1} W_{jq_2})B_{ij} + 2n^{-1}\sumi  \sumi  (W_{iq_1}-\E W_{iq_1}) \E W_{jq_2} B_{ij}. 
\end{align*}
Similarly, as Proof of \Cref{prop:bias-of-variance-estimator}, we can verify $\max_{i\in\mathcal{N}_n} |\hat{W}_{iq}-W_{iq}| = \Op(n^{-1/2})$, $\max_{i\in \mathcal{N}_n}\max\{|\hat{W}_{iq}|,|W_{iq}|\} = O(1)$, for any $1\leq q\leq Q+1$.
This leads to
\[
n^{-1}\sumi  \sumi  (\hat{W}_{iq_1}\hat{W}_{jq_2}-W_{iq_1} W_{jq_2})B_{ij} \leq \Op(n^{1/2}) n^{-1}\sumi \sumi B_{ij} = \Op(n^{1/2}) M_n(b_n,1) = \op(1).
\]
Similar as the proof of \cite[Theorem 4]{leung2022causal}, we have
\[
n^{-1}\sumi  \sumi  (W_{iq_1}-\E W_{iq_1}) \E W_{jq_2} B_{ij} = \op(1).
\]
Putting together, we have
\begin{align*}
    n^{-1}\sumi\sumj \hat{W}_{iq_1}\hat{W}_{iq_2}B_{ij} &= \Cov(n^{-1/2}\sumi {W}_{iq_1}, n^{-1/2}\sumi {W}_{jq_2}) + \\
    &\qquad \qquad \qquad n^{-1}\sumi\sumj \E {W}_{iq_1} \E {W}_{jq_2} B_{ij} + \op(1).
\end{align*}

As a consequence, we have
\begin{align*}
    &n^{-1}\sumi\sumj B_{ij} \hat{\bs{V}}_{\star,\phi(\bs{G}),i} \hat{\bs{V}}_{\star,\phi(\bs{G}),j}^\top = \Cov\Big(n^{-1/2}\sumi V_{\star,\phi(\bs{G}),i}\Big)+ \\
    &\qquad \qquad\qquad\qquad \frac{1}{n}\sumi \sumj B_{ij} (\bs{\tau}_{\phi(\bs{G}),i}-\bs{\tau}_{\phi(\bs{G})})(\bs{\tau}_{\phi(\bs{G}),j}-\bs{\tau}_{\phi(\bs{G})})^\top + \op(1),\\
    &n^{-1}\sumi\sumj B_{ij} \hat{\bs{V}}_{\star,\phi(\bs{G}),i} \hat{V}_{\star,j} = \Cov\Big(n^{-1/2}\sumi  V_{\star,\phi(\bs{G}),i},n^{-1/2}\sumi  V_{\star,i}\Big)+ \nonumber\\
    &\qquad \qquad\qquad\qquad \frac{1}{n}\sumi \sumj B_{ij} (\bs{\tau}_{\phi(\bs{G}),i}-\bs{\tau}_{\phi(\bs{G})})(\tau_j-\tau)^\top + \op(1).
\end{align*}
In light of the above, we have
$\tilde{\bs{\beta}}_{\star,\ND} = (\bs{H}_\star+\op(1))^{-1}(\bs{L}_\star+\op(1))$, where 
\begin{align*}
&\bs{L}_\star:=\Cov\Big(n^{-1/2}\sumi  V_{\star,\phi(\bs{G}),i},n^{-1/2}\sumi  V_{\star,i}\Big)+\frac{1}{n}\sumi \sumj B_{ij} (\bs{\tau}_{\phi(\bs{G}),i}-\bs{\tau}_{\phi(\bs{G})})(\tau_j-\tau)^\top\\
&\bs{H}_\star = \Cov\Big(n^{-1/2}\sumi V_{\star,\phi(\bs{G}),i}\Big)+\frac{1}{n}\sumi \sumj B_{ij} (\bs{\tau}_{\phi(\bs{G}),i}-\bs{\tau}_{\phi(\bs{G})})(\bs{\tau}_{\phi(\bs{G}),j}-\bs{\tau}_{\phi(\bs{G})})^\top.
\end{align*}

We see that $\tilde{\bs{\beta}}_{\star,\ND} = \bs{H}_\star^{-1}\bs{L}_\star$.

By Lemma~\ref{lem:order-of-haj-ht-mean} with $\bs{G}_i \equiv \phi(\bs{G}_i)$, we have
\begin{align*}
    \Cov\Big(n^{-1/2}\sumi V_{\star,\phi(\bs{G}),i}\Big) = O(1),\quad \Cov\Big(n^{-1/2}\sumi  V_{\star,\phi(\bs{G}),i},n^{-1/2}\sumi  V_{\star,i}\Big) = O(1).
\end{align*}
Combining with \Cref{a:bias-same-order-with-estimator-SM}, we have $\bs{H}_{\star} = O(1)$ and $\bs{L}_{\star} = O(1)$.

Let $\lambda_{\min}(\bs{H}_{\star})$ be the smallest eigenvalue of $\bs{H}_{\star}$. Assumption \ref{a:for-CLT-of-estimated-optimal-SM}~(i) suggest that $\liminf_{n\rightarrow \infty}\lambda_{\min}(\bs{H}_{\star})>0$ and therefore $(\bs{H}_{\star} + \op(1))^{-1} = \bs{H}_{\star}^{-1} + \op(1)$.

In light of the above, we have
\[
\hat{\bs{\beta}}_{\star,\ND} = \tilde{\bs{\beta}}_{\star,\ND} + \op(1).
\]

We are ready to prove prove (i). We see that
\begin{align*}
\hat{\tau}_\star(\hat{\bs{\beta}}_{\star,\ND})-\tau =& \hat{\tau}_\star(\hat{\bs{\beta}}_{\star,\ND})-\hat{\tau}_\star(\tilde{\bs{\beta}}_{\star,\ND})+\hat{\tau}_\star(\tilde{\bs{\beta}}_{\star,\ND})-\tau\\
    =&(\hat{\bs{\beta}}_{\star,\ND}-\tilde{\bs{\beta}}_{\star,\ND})^\top\{\hat{\bs{\mu}}_{\star,\phi(\bs{G})}(\bs{t})-\hat{\bs{\mu}}_{\star,\phi(\bs{G})}(\bs{t}^\prime)\}+\hat{\tau}_\star(\tilde{\bs{\beta}}_{\star,\ND})-\tau\\
    =& \op(1)\Op(n^{1/2}) + \hat{\tau}_\star(\tilde{\bs{\beta}}_{\star,\ND})-\tau \quad \text{(\Cref{lem:order-of-haj-ht-mean})} \\
    =& \hat{\tau}_\star(\tilde{\bs{\beta}}_{\star,\ND})-\tau +\op(n^{-1/2}).
\end{align*}
In light of the above, applying \cite[Theorem~2]{leung2022causal} with $Y_i\equiv Y_i-\phi(\bs{G}_i)^\top\tilde{\bs{\beta}}_{\HT,\ND}$, we have
 $n^{1/2}(\hat{\tau}_\HT(\hat{\bs{\beta}}_{\HT,\ND})-\tau)/\sigma_\HT(\tilde{\bs{\beta}}_{\HT,\ND})\xrightarrow{\textnormal{d}} \mathcal{N}(0,1)$. 
 On the other hand, we see that, $\tau_{\phi(\bs{G})} = \bs{0}$
\begin{align*}
    &\hat{\tau}_\haj(\tilde{\bs{\beta}}_{\haj,\ND})-\tau\\
    =& \hat{\tau}_\haj(\tilde{\bs{\beta}}_{\haj,\ND})-\tau - \tau_{\phi(\bs{G})}^\top \tilde{\bs{\beta}}_{\haj,\ND} \quad\quad\quad\quad \quad\quad (\tau_{\phi(\bs{G})} = \bs{0})\\
    =& \frac{\hat{\mu}_\HT(\bs{t})-\mu(\bs{t})\hat{1}_\HT(\bs{t})-\tilde{\bs{\beta}}_{\haj,\ND}^\top\{\hat{\bs{\mu}}_{\HT,\phi(\bs{G})}(\bs{t})-{\bs{\mu}}_{\phi(\bs{G})}(\bs{t})\hat{1}_\HT(\bs{t})\}}{\hat{1}_\HT(\bs{t})}-\\
    &\qquad
    \frac{\hat{\mu}_\HT(\bs{t}^\prime)-\mu(\bs{t}^\prime)\hat{1}_\HT(\bs{t}^\prime)-\tilde{\bs{\beta}}_{\haj,\ND}^\top\{\hat{\bs{\mu}}_{\HT,\phi(\bs{G})}(\bs{t}^\prime)-{\bs{\mu}}_{\phi(\bs{G})}(\bs{t}^\prime)\hat{1}_\HT(\bs{t}^\prime)\}}{1_\HT(\bs{t}^\prime)}\\
    =&\underset{=: T_1}{\underbrace{\Big[\hat{\mu}_\HT(\bs{t})-\mu(\bs{t})\hat{1}_\HT(\bs{t})-\tilde{\bs{\beta}}_{\haj,\ND}^\top\{\hat{\bs{\mu}}_{\HT,\phi(\bs{G})}(\bs{t})-{\bs{\mu}}_{\phi(\bs{G})}(\bs{t})\hat{1}_\HT(\bs{t})\}\Big]}}(1+\Op(n^{-1/2}))-\\
    &\qquad \underset{=:T_2}{\underbrace{\Big[\hat{\mu}_\HT(\bs{t}^\prime)-\mu(\bs{t}^\prime)\hat{1}_\HT(\bs{t}^\prime)-\tilde{\bs{\beta}}_{\haj,\ND}^\top\{\hat{\bs{\mu}}_{\HT,\phi(\bs{G})}(\bs{t}^\prime)-{\bs{\mu}}_{\phi(\bs{G})}(\bs{t}^\prime)\hat{1}_\HT(\bs{t}^\prime)\}\Big]}}(1+\Op(n^{-1/2}))
      \end{align*}
   Using \Cref{lem:order-of-haj-ht-mean} with $\bs{G}_i \equiv \phi(\bs{G}_i)$, we have $\hat{\mu}_\HT(\bs{t}^\prime)-\mu(\bs{t}^\prime)\hat{1}_\HT(\bs{t}) = \Op(n^{-1/2})$ and $\hat{\bs{\mu}}_{\HT,\phi(\bs{G})}(\bs{t})-{\bs{\mu}}_{\phi(\bs{G})}(\bs{t})\hat{1}_\HT(\bs{t}) = \Op(n^{-1/2})$, we have 
   \[
   \hat{\tau}_\haj(\tilde{\bs{\beta}}_{\haj,\ND})-\tau = T_1-T_2 + \op(n^{-1/2}).
   \]
   
Applying \cite[Theorem~2]{leung2022causal} with $Y_i\equiv Y_i-\mu(\bs{t})1(\bs{T}_i=\bs{t})-\mu(\bs{t}^\prime)1(\bs{T}_i=\bs{t}^\prime)-(\phi(\bs{G}_i)-\bs{\mu}_{\phi(\bs{G})}(\bs{t})1(\bs{T}_i=\bs{t})-\bs{\mu}_{\phi(\bs{G})}(\bs{t}^\prime)1(\bs{T}_i=\bs{t}^\prime))^\top\tilde{\bs{\beta}}_{\haj,\ND}$, we have $n^{1/2}(T_1-T_2)\xrightarrow{\textnormal{d}} \mathcal{N}(0,1)$. As a consequence, we have
 $n^{1/2}(\hat{\tau}_\haj(\hat{\bs{\beta}}_{\haj,\ND})-\tau)/\sigma_\haj(\tilde{\bs{\beta}}_{\haj,\ND})\xrightarrow{\textnormal{d}} \mathcal{N}(0,1)$. 

Now we prove \Cref{thm:estimator-with-optimal-estimated-precesion} (ii).
\begin{align*}
    &|\hat{\sigma}_{\star}^2(\hat{\bs{\beta}}_{\star,\ND})- \hat{\sigma}_{\star}^2(\tilde{\bs{\beta}}_{\star,\ND})|\leq  \Big|(\hat{\bs{\beta}}_{\star,\ND}-\tilde{\bs{\beta}}_{\star,\ND})^\top \Big\{n^{-1}\sumi\sumj B_{ij} \hat{\bs{V}}_{\star,\phi(\bs{G}),i} \hat{\bs{V}}_{\star,\phi(\bs{G}),j}^\top  \Bigr\}(\hat{\bs{\beta}}_{\star,\ND}-\tilde{\bs{\beta}}_{\star,\ND})\Big| + \\
    &\Big|n^{-1}\sumi\sumj B_{ij} (\hat{V}_{\star,i}-\hat{\bs{V}}_{\star,\phi(\bs{G}),i}^\top \tilde{\bs{\beta}}_{\star,\ND}) \hat{\bs{V}}_{\star,\phi(\bs{G}),j}^\top(\hat{\bs{\beta}}_{\star,\ND}-\tilde{\bs{\beta}}_{\star,\ND})\Big|.
\end{align*}
We have proven that
\begin{align*}
    &n^{-1}\sumi\sumj B_{ij} \hat{\bs{V}}_{\star,\phi(\bs{G}),i} \hat{\bs{V}}_{\star,\phi(\bs{G}),j}^\top = \bs{H}_\star + \op(1) = \Op(1)\\
    &n^{-1}\sumi\sumj B_{ij} \hat{\bs{V}}_{\star,\phi(\bs{G}),i} \hat{V}_{\star,j} = \bs{L}_\star + \op(1) = \Op(1).
\end{align*}
As a consequence, we have
\[
|\hat{\sigma}_{\star}^2(\hat{\bs{\beta}}_{\star,\ND})- \hat{\sigma}_{\star}^2(\tilde{\bs{\beta}}_{\star,\ND})| = \op(1)\Op(1)\op(1)+ \Op(1)\op(1) = \op(1).
\]

It is easy to verify that $\phi(\bs{G}_i) \equiv (\bs{X}_i^\top\bs{1}(\bs{T}_i = \bs{t}), \bs{X}_i^\top\bs{1}(\bs{T}_i = \bs{t}^\prime))^\top$, $\phi(\bs{G}_i) \equiv \phi_0(\bs{G}_i)$ satisfies \Cref{a:bounded-G-dc} for $\bs{G}_i$. Applying \Cref{thm:estimator-with-optimal-estimated-precesion-SM} with $\phi(\bs{G}_i) \equiv (\bs{X}_i^\top\bs{1}(\bs{T}_i = \bs{t}), \bs{X}_i^\top\bs{1}(\bs{T}_i = \bs{t}^\prime))^\top$, $\phi(\bs{G}_i) \equiv \phi_0(\bs{G}_i)$, \Cref{thm:estimator-with-optimal-estimated-precesion-Lin-Fisher} and \Cref{thm:estimator-with-optimal-estimated-precesion} follow.


\section{Constructive setting}\label{sec:counter}

In this section, we showcase an example where the efficiency gain of Lin's and Fisher's estimator proposed by \cite{gao2023causal} are negative. 
Consider the Bernoulli trial under the SUTVA setting, where $T_i \equiv D_i \in \{0,1\}$, $Y_i = \mu_i(1)D_i + \mu_i(0) (1-D_i)$. According to \cite[Lemma S1, S3]{gao2023causal}, let 
\[
\bs{\beta}_{\Lin}(t) := (\sum_{i=1}^n \bs{X}_i \bs{X}_i^\top)^{-1} \sum_{i=1}^n \bs{X}_i \mu_i(\tilde{t}), \quad \tilde{t}=0,1,\quad \bs{\beta}_{\Fisher} := \frac{1}{2}\Big(\bs{\beta}_{\Lin}(1) + \bs{\beta}_{\Lin}(0)\Big),
\]
be the asymptotic limit of $\hat{\bs{\beta}}_{\Lin}(\tilde{t})$ and $\hat{\bs{\beta}}_{\Fisher}$. The asymptotic variance of $n^{1/2}\hat{\tau}_{\Lin}$ and $n^{1/2}\hat{\tau}_{\Fisher}$ is then \citep[Theorem, 4.1, 4.4]{gao2023causal}:
\begin{align*}
&\sigma^2_{\Lin} = \frac{1}{n} \sum_{i=1}^n \pi_i(0)\pi_i(1)\Big(\frac{\mu_i(1)-\mu(1)}{\pi_i(1)} + \frac{\mu_i(0)-\mu(0)}{\pi_i(0)}- \bs{X}_i^\top\frac{\bs{\beta}_{\Lin}(1)}{\pi_i(1)}-\bs{X}_i^\top\frac{\bs{\beta}_{\Lin}(0)}{\pi_i(0)}\Big)^2,\\
&\sigma^2_{\Fisher} = \frac{1}{n} \sum_{i=1}^n \pi_i(0)\pi_i(1)\Big(\frac{\mu_i(1)-\mu(1)}{\pi_i(1)} + \frac{\mu_i(0)-\mu(0)}{\pi_i(0)}- \frac{\bs{X}_i}{\pi_i(1)\pi_i(0)}\bs{\beta}_{\Fisher}\Big)^2.
\end{align*}

For $n^{1/2}\hat{\tau}_{\haj}$, its asymptotic variance is instead
\[
\sigma^2_{\haj} = \frac{1}{n} \sum_{i=1}^n \pi_i(0)\pi_i(1)\Big(\frac{\mu_i(1)-\mu(1)}{\pi_i(1)} + \frac{\mu_i(0)-\mu(0)}{\pi_i(0)}\Big)^2.
\]

Let 
\begin{align*}
   \begin{pmatrix}
       \tilde{\bs{\beta}}(1) \\
       \tilde{\bs{\beta}}(0)
   \end{pmatrix} :=\begin{pmatrix}
       \sum_{i=1}^n \frac{\pi_i(0)}{\pi_i(1)} \bs{X}_i\bs{X}_i^\top & \sum_{i=1}^n  \bs{X}_i\bs{X}_i^\top\\
       \sum_{i=1}^n \bs{X}_i\bs{X}_i^\top & \sum_{i=1}^n \frac{\pi_i(1)}{\pi_i(0)} \bs{X}_i\bs{X}_i^\top
    \end{pmatrix}^{-1}\begin{pmatrix}
        \sum_{i=1}^n \frac{\pi_i(0)}{\pi_i(1)} \bs{X}_i\mu_i(1)+ \sum_{i=1}^n  \bs{X}_i\mu_i(0)\\
        \sum_{i=1}^n  \bs{X}_i\mu_i(1)+ \sum_{i=1}^n \frac{\pi_i(1)}{\pi_i(0)} \bs{X}_i\mu_i(0)
    \end{pmatrix},
\end{align*}
(assume that the denominator is invertible). We can construct the denominator to be invertible with varying propensity score. The efficiency gain of $n^{1/2}\hat{\tau}_{\Fisher}$ and $n^{1/2}\hat{\tau}_{\Lin}$ over $n^{1/2}\hat{\tau}_{\haj}$ can then be expressed as
\begin{align*}
   &\sigma_{\haj}^2 - \sigma_{\Lin}^2 =  \Delta - \frac{1}{n} \sum_{i=1}^n \pi_i(0)\pi_i(1)\Big(\bs{X}_i^\top\frac{\tilde{\bs{\beta}}(1)-\bs{\beta}_{\Lin}(1)}{\pi_i(1)}+\bs{X}_i^\top\frac{\tilde{\bs{\beta}}(0)-\bs{\beta}_{\Lin}(0)}{\pi_i(0)}\Big)^2, \\
    &\sigma_{\haj}^2 - \sigma_{\Fisher}^2 =  \Delta - \frac{1}{n} \sum_{i=1}^n \pi_i(0)\pi_i(1)\Big(\bs{X}_i^\top\frac{\tilde{\bs{\beta}}(1)-\bs{\beta}_{\Fisher}}{\pi_i(1)}+\bs{X}_i^\top\frac{\tilde{\bs{\beta}}(0)-\bs{\beta}_{\Fisher}}{\pi_i(0)}\Big)^2
\end{align*}
where
\begin{align*}
   \Delta = \frac{1}{n} \sum_{i=1}^n \pi_i(0)\pi_i(1)\Big(\bs{X}_i^\top\frac{\tilde{\bs{\beta}}(1)}{\pi_i(1)}+\bs{X}_i^\top\frac{\tilde{\bs{\beta}}(0)}{\pi_i(0)}\Big)^2.
\end{align*}
We can construct $\tilde{\bs{\beta}}(\tilde{t})=\bs{0}$, $\tilde{t}=0,1$, by finding an instance satisfying the following linear constraint
\begin{align}
\label{eq:construction-eq-1}
    \sum_{i=1}^n \frac{\pi_i(0)}{\pi_i(1)} \bs{X}_i\mu_i(1)+ \sum_{i=1}^n  \bs{X}_i\mu_i(0) = \bs{0},\quad 
        \sum_{i=1}^n  \bs{X}_i\mu_i(1)+ \sum_{i=1}^n \frac{\pi_i(1)}{\pi_i(0)} \bs{X}_i\mu_i(0) = \bs{0}.
\end{align}

 Now that $\tilde{\bs{\beta}}(\tilde{t})=0$, $\tilde{t}=0,1$, we further have
\begin{align*}
  \sigma_{\haj}^2 - \sigma_{\Lin}^2 = &-\frac{1}{n} \sum_{i=1}^n \pi_i(0)\pi_i(1)\Big(\bs{X}_i^\top\frac{\bs{\beta}_{\Lin}(1)}{\pi_i(1)}+\bs{X}_i^\top\frac{\bs{\beta}_{\Lin}(0)}{\pi_i(0)}\Big)^2\\
  =  & -\begin{pmatrix}
       \bs{\beta}_{\Lin}(1) \\
       \bs{\beta}_{\Lin}(0)
\end{pmatrix}^\top\begin{pmatrix}
       \sum_{i=1}^n \frac{\pi_i(0)}{\pi_i(1)} \bs{X}_i\bs{X}_i^\top & \sum_{i=1}^n  \bs{X}_i\bs{X}_i^\top\\
       \sum_{i=1}^n \bs{X}_i\bs{X}_i^\top & \sum_{i=1}^n \frac{\pi_i(1)}{\pi_i(0)} \bs{X}_i\bs{X}_i^\top
    \end{pmatrix}
    \begin{pmatrix}
       \bs{\beta}_{\Lin}(1) \\
       \bs{\beta}_{\Lin}(0).
\end{pmatrix} 
\end{align*}
This means that by further restricting 
\begin{align}
    \label{eq:construction-eq-2}
    \sum_{i=1}^n \bs{X}_i \mu_i(\tilde{t}) \ne \bs{0}, \quad \tilde{t} =0,1,
\end{align}
${\bs{\beta}}_{\Lin}(\tilde{t}) \ne \bs{0}$ and since the middle matrix is invertible, we have $\sigma_{\haj}^2 - \sigma_{\Fisher}^2$ is strictly negative.

Similarly, by restricting
\begin{align}
\label{eq:construction-eq-3}
\sum_{i=1}^n \bs{X}_i \mu_i(1) + \sum_{i=1}^n \bs{X}_i \mu_i(0) \ne \bs{0},
\end{align}
$\sigma_{\haj}^2 - \sigma_{\Lin}^2$ is strictly negative.

In summary, we have we can construct \eqref{eq:construction-eq-1}--\eqref{eq:construction-eq-3}, such that the efficiency gain of Lin's regression and Fisher's regression are negative.

\subsection{An numerical experiment inspired by the counter-example}
\label{sec:additional-simulation}

In this subsection, we conduct a numerical experiment that approximately satisfies the counter-example shown above. We consider the SUTVA setting with $n = 1000$ experimental candidates. We follow the generation process:
\begin{equation}
    \begin{aligned}
    {Y}_i = \bs{X}_i\left(-(1 - D_i) + \frac{\pi_i(1)}{\pi_i(0)} \cdot D_i\right) + \varepsilon_i. 
    \end{aligned}\label{generation}
\end{equation}
Here $\varepsilon_i$ and $\bs{X}_i$ are i.i.d sampled from $\mathcal{N}(0, 0.25)$, and then propensity score $\{\pi_i(1)\}_{i \in \mathcal{N}_n}$ is uniformly sampled from $[0.1,0.9]$. As in the main paper, we sample them just in the first draw then keep them fixed throughout the experiment. To see why such data generation process can approximately satisfy~\ref{eq:construction-eq-1}--\ref{eq:construction-eq-3}, observe that
\begin{align*}
& \frac{1}{n} \sum_{i=1}^n \frac{\pi_i(0)}{\pi_i(1)} \bs{X}_i\mu_i(1)+ \frac{1}{n} \sum_{i=1}^n  \bs{X}_i\mu_i(0) \\
&\quad = \frac{1}{n} \sum_{i=1}^n \frac{\pi_i(0)}{\pi_i(1)} \bs{X}_i \left(\frac{\pi_i(1)}{\pi_i(0)} \bs{X}_i + \varepsilon_i\right) + \sum_{i=1}^n \bs{X}_i \left(-\bs{X}_i + \varepsilon_i\right) \\
&\quad = \frac{1}{n} \sum_{i=1}^n \bs{X}_i \varepsilon_i \left(\frac{\pi_i(0)}{\pi_i(1)} 
 + 1\right),
\end{align*}
which is approximately zero as $n \to \infty$. Using an analogous argument it is also straightforward to verify that the second term of~\eqref{eq:construction-eq-1} can be satisfied approximately. To see how~\eqref{eq:construction-eq-2} and~\eqref{eq:construction-eq-3} can be satisfied, we notice that
\begin{align*}
    & \frac{1}{n} \sum_{i=1}^n \bs{X}_i \mu_i(1) = \frac{1}{n} \sum_{i=1}^n \bs{X}_i^2 \frac{\pi_i(1)}{\pi_i(0)} + \frac{1}{n} \sum_{i=1}^n \bs{X}_i\varepsilon_i \approx \frac{1}{n} \sum_{i=1}^n \bs{X}_i^2 \frac{\pi_i(1)}{\pi_i(0)} \approx 0.437 \\
    & \frac{1}{n} \sum_{i=1}^n \bs{X}_i \mu_i(0) = - \frac{1}{n} \sum_{i=1}^n \bs{X}_i^2 + \frac{1}{n} \sum_{i=1}^n \bs{X}_i\varepsilon_i \approx - \frac{1}{n} \sum_{i=1}^n \bs{X}_i^2 \approx -0.250.
\end{align*}

We report ``$\text{ND-F}, \text{ND-L}$'' mentioned in our main text. Moreover, we set the general auxiliary variables $\bs{G}_i \equiv (\bs{X}_i^\top \bs{1}(T_i=1), \bs{X}_i^\top \bs{1}(T_i=0))^\top$ and report it as $\text{ND-}\phi_0(\text{G})$ with the normalizing procedure~\eqref{eq:point-wise-decorrelation}. 

We compare our methods with ``HT, Haj'' and Fisher's and Lin's regression adjusted estimator (``F, L'') in~\citet{gao2023causal}, respectively. We conduct $10^5$ random assignments. Table~\ref{counter_tab} illustrates that ``F, L'' are obviously less efficient than the unadjusted Hajek (``Haj'') estimator, with a decrease of $10.81\%, 35.14\%$, respectively. However, Our method ``$\text{ND-}\text{F}^{}$, $\text{ND-}\phi_0(\text{G}^{\text{}})$'' holds approximately the same efficiency as ``Haj''. Moreover, the ``ND-L'' estimator also harms the efficiency of the ``Haj'', which validates the necessity of normalizing procedure~\eqref{eq:point-wise-decorrelation}.


\begin{table}[t]
\scalebox{1}{
\begin{tabular}{c cccc ccc}
\hline\hline Outcome model & \multicolumn{7}{c}{ Counter example } \\
\hline Method & $\text{HT}$ & $\text{Haj}$ & ${\text{F}}$  & $\text{L}_{\text{}}$ & ${\text{ND-}\text{F}}$ & $\text{ND-L}$ & $\text{ND-}\phi_0(\text{G})$  \\
\hline
Empirical absolute bias & 0.001 & 0.001 & 0.000 & 0.000 & \textbf{0.000} & \textbf{0.000} & \textbf{0.001} \\
\hline 
Oracle SE & 0.037 &0.037 & 0.041 &0.050 &\textbf{0.037} &     \textbf{0.043} & \textbf{0.038}\\
Estimated SE & 0.066 & 0.066 &0.068 &0.058&    \textbf{ 0.066}  &       \textbf{0.054} & \textbf{0.066}\\
\hline 
Oracle coverage probability & 0.950 & 0.950 
& 0.950 
& 0.950 
&\textbf{0.950}
& \textbf{0.950} & \textbf{0.950}
\\
Empirical coverage probability & 0.999&0.999&0.999&0.978 &\textbf{0.999}&\textbf{0.986} & \textbf{0.999} \\
\hline\hline
\end{tabular} }
\caption{The counter-example. $\tau = 0.024$.} \label{counter_tab}
\end{table}



\section{Additional details of the data-generation process in the simulation study}
\label{sec:details-gen}

In this section, we provide additional details regarding the data-generation process of the $Y_i$'s in both the linear and nonlinear models. For the linear-in-means model (i), the generation model is naturally equivalent to the following closed form:
$$\bs{Y} =  (\bs{I}_n - \alpha_1 \bs{O})^{-1} [ \alpha_0 \bs{1}_n + (\alpha_2 \bs{O} +\alpha_3 \bs{I}_n)\bs{D} + {\alpha}_4 \bs{X} + \bs{\varepsilon}]. $$
   Here $\bs{I}_n$ is the $n-$dimensional identity matrix, and $\bs{1}_n$ is an all-one vector. 
Moreover, $\bs{\varepsilon} = (\varepsilon_1,\ldots,\varepsilon_n)^\top$ and $\bs{O}$ is row-normalized from $\bs{A}$, i.e., $\bs{O} = (A_{ij}/\sum_{l=1}^n A_{il})_{i,j\in \mathcal{N}_n}$. On the other hand, For the nonlinear contagion model (ii), $\bs{Y}$ is generated via the following iteration:
\begin{equation*}
    \begin{aligned}
    Y_i^t= \bs{1} \Big( \alpha_0 + \alpha_1 \frac{\sum_{j=1}^n A_{i j} Y^{t-1}_j}{\sum_{j=1}^n A_{i j}}+\alpha_2 \frac{\sum_{j=1}^n A_{i j} D_j}{\sum_{j=1}^n A_{i j}}+\alpha_3 D_i+ {\alpha_4} \bs{X}_i+\varepsilon_i > 0 \Big), ~ t \geq 1.
    \end{aligned}
\end{equation*} The initialization step is defined as $Y_i^0 = 0$. Such iteration continues until it achieves the convergence $Y_i^t = Y_i^{t-1}, \forall i \in \mathcal{N}_n $, then we finally set $Y_i = Y_i^t$.

\end{document}